\documentclass[11pt, oneside]{article}   	
\usepackage[margin=2cm]{geometry}                		
\usepackage[parfill]{parskip}    		
\usepackage{lipsum}	
\usepackage{graphicx}			
\usepackage{wrapfig}
\usepackage[utf8]{inputenc}
\usepackage[english]{babel} 
\usepackage{amsmath,amssymb,amsthm,bm} 
\usepackage[margin=2cm]{geometry}
\usepackage[color=yellow]{todonotes}
\usepackage{url,multicol}	
\usepackage{esint}
\usepackage[colorlinks]{hyperref}
\usepackage{enumerate}
\usepackage{outlines}[enumerate]
\usepackage{framed,comment}
\usepackage{upgreek}
\usepackage{pgfplots}
\usepackage{float}
\usepackage{tikz,tikz-cd}
\usepackage{rotating}
\usepackage{caption}
\usepackage{cancel}
\usepackage{subcaption}
\usepackage{appendix}
\usepackage{tikz-cd}
\usepackage{pgf, pgffor}
\extrafloats{100}
\numberwithin{equation}{section}

\newtheorem{theorem}{Theorem}[section]
\newtheorem{lemma}{Lemma}[theorem]

\newtheorem{proposition}{Proposition}[theorem]
\newtheorem{remark}{Remark}[theorem]

\newcommand{\wt}[1]{\widetilde{#1}}
\newcommand{\wh}[1]{\widehat{#1}}

\newcommand{\mbf}[1]{\mathbf{#1}}
\newcommand{\mc}[1]{\mathcal{#1}}
\newcommand{\bs}[1]{\boldsymbol{#1}}

\newcommand{\dd}{\mathrm{~d}}

\newcommand{\dt}{\, \mathrm{d}t}
\newcommand{\E}[1]{\mathbb{E}\left[{#1}\right]}

\def\p{{\partial}}

\def\bk{{\mathbf{k}}}
\def\bu{{\mathbf{u}}}
\def\bx{{\mathbf{x}}}

\newcommand{\expect}{\mathbb{E}}

\def\p{\partial}

\begin{document}
\title{Deterministic and Stochastic Euler-Boussinesq Convection}
\author{\bigskip\Large
Darryl D. Holm and Wei Pan 
\\ \bigskip\small
Mathematics, Imperial College London SW7 2AZ, UK\\
\small Email: d.holm@ic.ac.uk and wei.pan@imperial.ac.uk}

\date{Fondly remembering Charlie Doering's ``obsession with convection''.}

\maketitle

\makeatother

\begin{abstract}
Stochastic parametrisations of the interactions among disparate scales of motion in fluid convection are often used for estimating prediction uncertainty, which can arise due to inadequate model resolution, or incomplete observations, especially in dealing with atmosphere and ocean dynamics, where viscous and diffusive dissipation effects are negligible. 
This paper derives a family of three different types of stochastic parameterisations for the classical Euler-Boussinesq (EBC) equations for a buoyant incompressible fluid flowing under gravity in a vertical plane. These three stochastic models are inspired by earlier work on the effects of stochastic fluctuations on transport, see, e.g.,  \cite{Kraichnan1968,Kraichnan1994,Doering1994}. They are derived here from variants of Hamilton's principle for the deterministic case when Stratonovich noise is introduced. The three models possess different variants of their corresponding Hamiltonian structures. One variant (SALT) introduces stochastic transport. Another variant (SFLT) introduces stochastic \emph{forcing}, rather than stochastic \emph{transport}. The third variant (LA SALT) introduces nonlocality in its stochastic transport, in the probabilistic sense of McKean \cite{McKean1966}. 
\end{abstract}

\tableofcontents

\color{black}

\section{Introduction}

\subsection{Story line}
To derive the three stochastic parametrisations of EBC treated here, we first obtain the Lie-Poisson Hamiltonian formulation of the deterministic EBC equations, in which all dimension-free ratios of units can be conveniently absorbed into the temperature, or buoyancy, variable. We then introduce stochasticity into the EBC equations in three different ways which all preserve their Hamiltonian structure and their Kelvin circulation theorem. The basis of our approach is SALT (Stochastic Advection by Lie Transport) \cite{Holm2015}, a method that has recently been applied successfully in \cite{cotter2018modelling, cotter2019numerically} as an algorithm for stochastically calibrating data, dynamically estimating uncertainty and using particle filtering methods for data assimilation to reduce uncertainty in confined 2D Euler flows, in quasigeostrophic channel flows and in thermal quasigeostrophic flows.  The SALT algorithm aims to quantify the uncertainty that arises, e.g., in the process of up-scaling, or coarse-graining of either observed or synthetic data at fine scales, for use in computational simulations at coarser scales. The SALT approach also generalises to accommodate coloured noise, fractional Brownian motion, and other subsets of the geometric rough paths, \cite{CHLN2022}.      

The paper begins by discussing the initial value problem for convective adjustment of an \emph{ideal} stratified fluid out of equilibrium in a vertical plane, under the influence of gravity, in the absence of dissipation by viscosity and thermal diffusivity. The motion is described by the planar Euler-Boussinesq convection (EBC) equations in the vorticity representation. This is a non-dissipative, initial-value version of the dissipative Oberbeck-Boussinesq convection studied in \cite{Farhat-etal-RBCdata2020}. The EBC equations may be derived from Hamilton's variational principle. The corresponding Hamiltonian structure of the EBC model may be used, for example, to characterise the equilibrium solutions of the non-dissipative system as critical points of the sum of the Hamiltonian on level sets of the conserved quantities and thereby derive the corresponding Hamiltonian Taylor-Goldstein equations governing their linear dynamics. 

Following earlier treatments of the effects of fluctuations on transport in physics \cite{Kraichnan1968,Kraichnan1994,Doering1994} as surveyed in \cite{Frisch1995}, and in mathematics as surveyed in \cite{FlandoliPappa2022}, we introduce \emph{stochastic advection by Lie transport} (SALT) in the Hamiltonian framework for EBC, thereby recovering the stochastic EBC equations studied in \cite{Alonso-Oran-etal-JStaPhys,Holm2015} which were derived from the Euler-Poincar\'e variational approach. Following the approach of \cite{cotter2018modelling,cotter2019numerically} as developed for quasigeostrophic channel flow, the SALT EBC equations enable uncertainty quantification and also admit data assimilation methods based on particle filtering that can reduce the uncertainty in coarse-grained computational simulations of convective adjustment.%
\footnote{Data calibration and assimilation methods for SALT EBC will be discussed elsewhere.} 
Next, we discuss two variants of SALT for the stochastic EBC initial value problem. These are: (1) the method of stochastic \emph{forcing} by Lie transport (SFLT); and (2) the \emph{Lagrangian Averaged} SALT (LA SALT) equations. 

The SFLT and LA SALT approaches yield equations that are similar to classical models of convective adjustment. In particular, the SFLT EBC equations are reminiscent of Landau's fluctuation-dissipation theory of hydrodynamics \cite{Landau}. Moreover, the expected equations for the LA SALT EBC model turns out to be quite similar in appearance to the classical Oberbeck-Boussinesq equations for diffusive, viscous dynamics of the classical Rayleigh-B\'enard convection problem. Conveniently, the LA SALT EBC fluctuation equations turn out to be linear. This linearity of the fluctuation dynamics enables the derivation of  deterministic dynamical equations for the covariances and higher moments of the fluctuations away from the expected solution of the LA SALT EBC equations. This progression of results is reminiscent of the classic work of \cite{Kraichnan1968} for stochastic passive transport, although of course here we have restricted to a stochastic model of the planar EBC equations.

In summary, the present work begins by discussing the solution properties of the deterministic EBC equations. In particular, the Hamiltonian structure of the EBC equations and the linear stability analysis of their equilibria are discussed, following \cite{AH1987,AHMR1986,HMRW1985}.
We then transfer as much as possible of this deterministic structure to the stochastic EBC dynamics by applying the SALT approach introduced in \cite{Holm2015}. Specifically, the SALT approach is adapted here to the Hamiltonian formulation for the deterministic EBC equations. The analysis of the solution behaviour of the SALT EBC equations has been treated in the literature, particularly \cite{Alonso-Oran-bDLeon-JNLS2020,CHLMP2021,GHL2019}. The present work also makes contact with  two variants of the SALT algorithm. These are the SFLT variant  and the LA SALT variant. The SFLT variant \cite{HH2021} is reminiscent of Landau's fluctuating hydrodynamics approach, \cite{Landau}.  The LA SALT variant for the EBC initial value problem whose derivation and well-posedness was studied in \cite{Alonso-Oran-etal-JStaPhys}.

\paragraph{\bf Main content of the paper}
\begin{enumerate}
    \item 
    Section \ref{sec-ModBkgrd} explains the geometric mechanics background underlying the Hamiltonian dynamics of EBC. Geometrically, the dynamics of EBC is understood as coadjoint Hamiltonian motion generated by the semidirect-product action of the Lie algebra $(f_1\circledS f_2)$ on function pairs $(f_1;f_2)\in (f_1\circledS f_2)^*$. 
    
    This coadjoint Hamiltonian motion is a fruitful geometric framework for the application of the SALT approach to stochastic modelling in fluid dynamics. 
    \item 
    Section \ref{sec: SALT-EBC} introduces the mathematical preliminaries and notations for the stochastic analysis of SALT EBC.
    
    The purpose of introducing the SALT equations is to enable uncertainty quantification for EBC and thereby admit data assimilation methods following \cite{cotter2018modelling,cotter2019numerically} that enable the reduction of uncertainty in coarse-grained computational simulations of convective adjustment in a vertical plane. No details are given concerning the application of the SALT equations in data analysis and assimilation methods for EBC, because this is work currently in progress. For an example of such applications, the reader may refer to the data analysis and assimilation results for the related case of quasigeostrophic channel flow in  \cite{cotter2018modelling,cotter2019numerically}.

In section \ref{sec:SEBC-2variants}, we discuss two variants of SALT EBC.   These are: (1) the method of stochastic \emph{forcing} by Lie transport (SFLT); and (2) the \emph{Lagrangian Averaged} SALT (LA SALT) equations. 

As mentioned earlier, the SFLT EBC equations are reminiscent of Landau's fluctuation-dissipation theory of hydrodynamics \cite{Landau}. The SFLT EBC equations were introduced in \cite{HH2021} and have been applied to the three-dimensional hydrostatic primitive equations for a rotating stratified fluid \cite{HuPatching2022}. Hence, the derivation and application of SFLT EBC will not be discussed in detail here, because this is also work in progress. 
    
    Lagrangian Averaged SALT (LA SALT) equations for the EBC  initial value problem, whose properties were studied in \cite{Alonso-Oran-etal-JStaPhys}. The  expected equations of the LA SALT EBC  turn out to be quite similar in appearance to the original dissipative Oberbeck-Boussinesq equations for RBC. 
The LA SALT equations for EBC lead to deterministic dynamical equations for the covariances and higher moments of the fluctuations away from the expected solution of the LA SALT EBC equations in the vorticity representation. As discussed further in \cite{Alonso-Oran-etal-JStaPhys}, the evolution of these statistical properties for LA SALT EBC provides a glimpse into how one might imagine how to define climate change in terms of its evolving statistics, rather than being simply the long-term dynamics of a 30-year time average. 

   \item
    Section \ref{sec: Conclusion} provides a brief conclusion and indicates several other promising directions for future research in the stochastic analysis of EBC.  
    \item
    Appendix \ref{appendix-A} discusses the Lie-Poisson Hamiltonian structure of deterministic EBC. The Hamiltonian structure of EBC guides the study of stability of its linearised solutions as well as its stochastic generalisations. 
    \item
    Appendix \ref{appendix-B} discusses the class of EBC equilibria which are critical points of a certain constrained energy. We then study 
their linear Lyapunov stability conditions using the energy-Casimir approach and determine their linear instability properties by deriving the EBC version of the Taylor-Goldstein equation for stratified incompressible planar fluid flow. 
\end{enumerate}

\subsection{Problem formulation}

The classical Rayleigh-B\'enard convection (RBC) equations in the Boussinesq approximation describe creation of circulation due to misalignments of gradients of buoyancy from the vertical direction of gravity in a planar slice. Under gravity, the motion converts buoyant potential energy derived from heating below (and cooling above) into circulation and kinetic energy as the fluid rises. 
The augmentation of Euler's fluid equations to account for the circulation effects of vertical buoyancy gradients under gravity is represented by the dynamics of Rayleigh-B\'enard convection (RBC).

The fluid motion considered in the classical case takes place on a 2D planar rectangular vertical slice with horizontal coordinate $x'=x/L$ and vertical coordinate $z'=z/H$. The bottom edge at $z/H=0$ is heated to a certain constant temperature $\Theta_{z=0}$ and the top edge at $z/H=1$ is cooled to a lesser constant temperature $\Theta_{z=H}$. The divergence-free fluid flow velocity $\bu(x',z',t)=: \nabla^\perp\psi = (-\,\psi_{z'}\,,\, \psi_{x'})$ with stream function $\psi(x',z',t)$ in the vertical slice is usually assumed to obey the Navier-Stokes equations under the Oberbeck-Boussinesq approximation. 

The fluid is assumed to have a constant heat capacity, so the advection-diffusion equation for heat can be expressed in terms of its temperature, $\Theta$. The governing set of equations for the classical Rayleigh-B\'enard convection problem under these assumptions is given by \cite{Saltzman1962} as,
\begin{align}
\begin{split}
\frac{\partial}{\partial t}\textbf{u} + \textbf{u}\cdot\nabla \textbf{u} &= -\nabla p + \nu \Delta \textbf{u} + \textbf{F} ,\\
\frac{\partial}{\partial t} {\Theta} + \textbf{u}\cdot \nabla {\Theta} &= \kappa\Delta {\Theta},\\
\nabla\cdot \textbf{u} &= 0.
\end{split}
\label{DOBeqns}
\end{align}
In the top equation in \eqref{DOBeqns}, the pressure $p$ enforces incompressibility and the buoyancy force is given by 
\begin{equation}
\mathbf{F}  = g b \hat{\textbf{z}} = \alpha g ({\Theta}/{\Theta}_{ref}-1) \hat{\textbf{z}}
\quad\hbox{with}\quad
1+b = \rho/\rho_{ref}= 1 - \alpha {\Theta}/{\Theta}_{ref}
\,.
\label{Buoy-Force}
\end{equation}
The buoyancy force \textbf{F} is assumed to act in the vertical direction and depend on the (positive) thermal expansion coefficient $\alpha>0$, the acceleration of gravity $g$ and the deviation of the local temperature ${\Theta}$ from a time-independent reference temperature ${\Theta}_{ref}$. The dissipation terms involve the linear diffusion of momentum per unit mass $\textbf{u}$, by viscosity $\nu$, and heat per unit mass ${\Theta}$, by heat diffusivity $\kappa$. In the absence of these dissipation terms in $\nu$ and $\kappa$, the dynamics is reversible and Hamiltonian.

\section{Hamiltonian modelling background}\label{sec-ModBkgrd}


\subsection{Lie-Poisson Hamiltonian Euler-Boussinesq convection (EBC) equations} 

\paragraph{\bf Units}
In the absence of dissipation, $\nu=0=\kappa$, the Rayleigh-B\'enard convection (RBC) equations in \eqref{DOBeqns} become the  Euler-Boussinesq convection (EBC) equations, given by 
\begin{align}
\begin{split}
\frac{\partial}{\partial t}\textbf{u} + \textbf{u}\cdot\nabla \textbf{u} &= -\nabla p + \alpha g ({\Theta}/{\Theta}_{ref}-1) \hat{\textbf{z}}\,,\\
\frac{\partial}{\partial t} {\Theta} + \textbf{u}\cdot \nabla {\Theta} &= 0,\\
\nabla\cdot \textbf{u} &= 0.
\end{split}
\label{DEBCeqns}
\end{align}

In the EBC equations, the natural units of measure are: (i) length $\sqrt{x^2+z^2}$ in units of the height of the domain $H$;  (ii) velocity, $\bu$, in units  $U$ (ii) time, $t$, in units of $H/U$; (iv) stream function, $\psi$ in units of $UH$; (v) and local temperature, $\Theta$, in units of $\Theta_{ref}$. 

Taking the curl of the EBC motion equation in \eqref{DEBCeqns} and introducing the Froude number $Fr=U/\sqrt{gH}$ in these units of space and time yields the dimension-free equations,
\begin{align}
\begin{split}
\partial_t \omega + \bm{u}\cdot \nabla \omega
&= \frac{\alpha}{Fr^2}\Theta_x
\,,
\\
\partial_t \Theta + \bm{u}\cdot \nabla \Theta &= 0
\,,\end{split}
\label{EBC-eqns-def1}
\end{align}
\begin{align}
\begin{split}
\quad\hbox{with}\quad 
\bm{u} &= -\,\bs{\hat{y}}\times \nabla\psi =: \nabla^\perp\psi = (-\,\psi_z\,,\, \psi_x)\,,
\quad\hbox{so}\quad
{\rm div}\bm{u} = 0
\,,\\
\quad\hbox{and}\quad 
\omega &:= -\,\bs{\hat{y}}\cdot {\rm curl}\bm{u} = \Delta \psi  = \psi_{xx} + \psi_{zz}
\quad\hbox{so}\quad 
-\,\bs{\hat{y}}\cdot {\rm curl}(\Theta \hat{\textbf{z}}) = \Theta_x
\,.
\end{split}
\end{align}
Here, the quantity $\omega:=-\,\bs{\hat{y}}\cdot {\rm curl}\bm{u} = \Delta \psi$ is the $\bs{\hat{x}}\times\bs{\hat{z}}=-\,\bs{\hat{y}}$ component of the vorticity, pointing \emph{into the $(x,z)$ plane}, to follow the notation in \cite{Saltzman1962}.

\paragraph{\bf Jacobian notation}
The misalignment of the gradients $\nabla a$ and $\nabla b$ of functions $a,b$ in a vertical (x,z) plane is expressed by the Jacobian, denoted
\begin{align}
J(a,b):=-\,\bs{\hat{y}}\cdot \nabla a\times \nabla b = a_x b_z-a_z b_x
\quad\hbox{or equivalently,}\quad
J(a,b)dx\wedge dz = da\wedge db\,.
\label{def-Jac}
\end{align}

%
In terms of the Jacobian and the local temperature $\Theta(x,z,t)$, the EBC equations may be expressed as
\begin{align}
\begin{split}
\partial_t \omega + J(\psi,\omega) &= -\, \frac{\alpha}{Fr^2}J(z,\Theta)  \,,
\\
\partial_t \Theta + J(\psi,\Theta ) &= 0
\,.
\end{split}
\label{EBC-eqns-def2}
\end{align}

\paragraph{\bf Kelvin circulation theorem for EBC}
It follows immediately from equation \eqref{EBC-eqns-def1} for velocity $\bu=\nabla^\perp\psi =  (-\,\psi_z\,,\, \psi_x)$ and  vorticity $\omega=\Delta \psi= \psi_{xx}+\psi_{zz}$ that 
\begin{align}
\frac{d}{dt}\oint_{c(u)} \!\!\!\! \mbf{u}\cdot d\bx
= 
\frac{d}{dt}\int\!\!\!\!\int_{\p S = c(u)} \!\!\!\! \omega \,dxdz
=  \frac{\alpha}{Fr^2}\int\!\!\!\!\int_{\p S = c(u)} \!\!\!\!  dz\wedge d{\Theta}
\,,\label{EBC-KelvinThm}
\end{align}
for a closed material loop, $c(u)$, moving with the flow velocity $\bu=\nabla^\perp\psi$. Thus, EBC is driven by misalignment between the gradients $\nabla z$ and $\nabla {\Theta}$. 

\paragraph{\bf Rescaling to simplify notation}
Hereafter, we will absorb $\alpha/Fr^2$ into $\Theta$ as $\Theta
\to \Theta':=(\alpha/Fr^2)\Theta$, and then simplify the notation by dropping the prime.

In this simplified notation, the EBC equations \eqref{EBC-eqns-def2} become
\begin{align}
\begin{split}
\partial_t \omega + J(\psi,\omega) + J(z,\Theta) &= 0\,,
\\
\partial_t \Theta + J(\psi,\Theta ) &= 0
\,.
\end{split}
\label{EBC-eqns-def3}
\end{align}
The vertical boundary conditions are $\bu\cdot\bs{\wh{n}} = 0 $ and $\bs{\wh{n}}\times \nabla \Theta = 0$. Equivalently, one may take  $\psi\big|_{z=0}=0$, $\psi_{z=H}=0$ and $\Theta\big|_{z=0}=\Theta_0>0$,  $\Theta\big|_{z=H}=0$; and periodic in the horizontal coordinate $x$, see \cite{Saltzman1962}. 

\paragraph{\bf Examples of computational simulations} Snapshots of three example solutions of the initial value problem for the EBC equations \eqref{EBC-eqns-def2} appear in Figures \ref{figure: example1}--\ref{figure: example3}, below. 
These examples are inspired by the process of oceanic deep water formation driven by Rayleigh-Taylor instability which produces rapidly down-welling water columns called \emph{chimneys} when warm ocean currents run into cold Arctic seas; such as, when the Gulf Stream runs into the Greenland Sea \cite{Kovalevsky2020}. The EBC equations \eqref{EBC-eqns-def2} apply in such oceanic processes, since viscosity and thermal diffusivity can be neglected because of the immense inertia involved. 

The domain in each of these three figures is horizontally periodic and has fixed boundaries above and below. 
In the computational simulations, we used a mixed finite element scheme for the spatial derivatives, and SSPRK3 scheme for the time stepping. The domain was chosen to be a square channel $[0,1]^2$ with $512^2$ resolution. The boundaries were periodic in $x$ and had ``fixed walls" at the bottom and top, along $z=0$ and $z=1$, respectively. We prescribed no-slip boundary conditions, i.e. $\psi=0$ on the fixed boundaries. 

Figure \ref{figure: example1} depicts snapshots of the evolution of an unstably stratified (Rayleigh-Taylor) density distribution superposed initially on a periodic array of Langmuir circulations. Figures \ref{figure: example2} and \ref{figure: example3} depict, respectively, the dynamics of an initially unstably stratified (Rayleigh-Taylor) density distribution that is perturbed along the bottom of the domain slightly to the right of its centre by a small Gaussian-shaped perturbation of temperature (in Figure \ref{figure: example2}), or of vorticity, (in Figure \ref{figure: example3}). The dynamics of all three example solutions start slowly and then accelerate as the stretching and folding of the flow develops sharp thermal shear fronts that then roll up into high wave-number instabilities and eventually mix into a sort of equipartitioned-looking `soup' of small scale fluctuations. 

Animations of these examples of the initial value problem for the EBC equations appear online at the following websites:\\
Example 1: \url{https://youtu.be/zorhwJ0pmUI}\\
Example 2: \url{https://youtu.be/pXU5mJqQjuA}\\
Example 3: \url{https://youtu.be/FFdxxyyRVk8}


\begin{figure}[h!]
\centering
\begin{subfigure}[b]{0.625\textwidth}
\centering
\includegraphics[width=\textwidth, trim=102 0 124 93, clip]{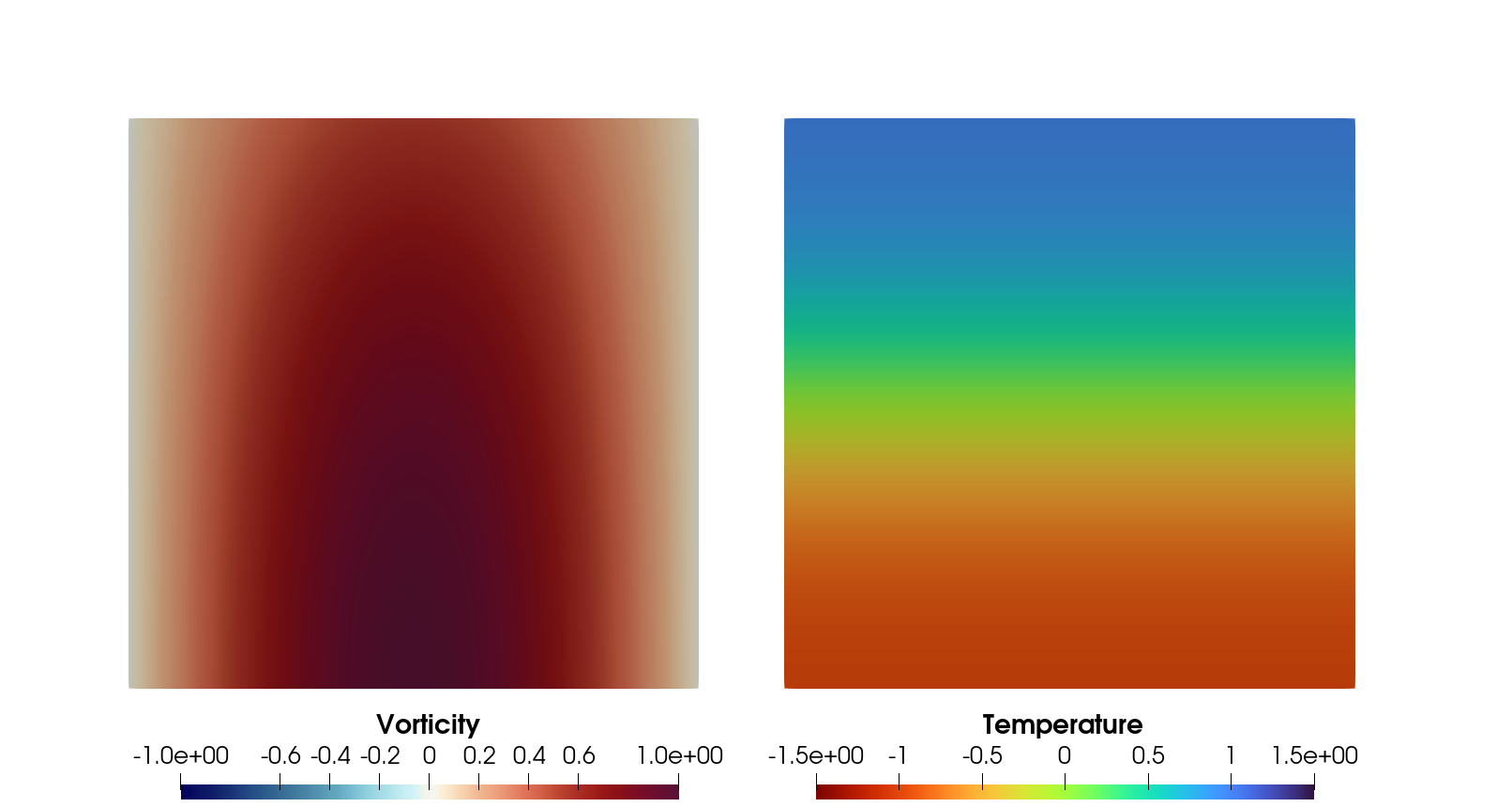}
\label{fig:example 1 - a}
\caption{}
\end{subfigure}
\begin{subfigure}[b]{0.625\textwidth}
\centering
\includegraphics[width=\textwidth, trim=102 0 124 93, clip]{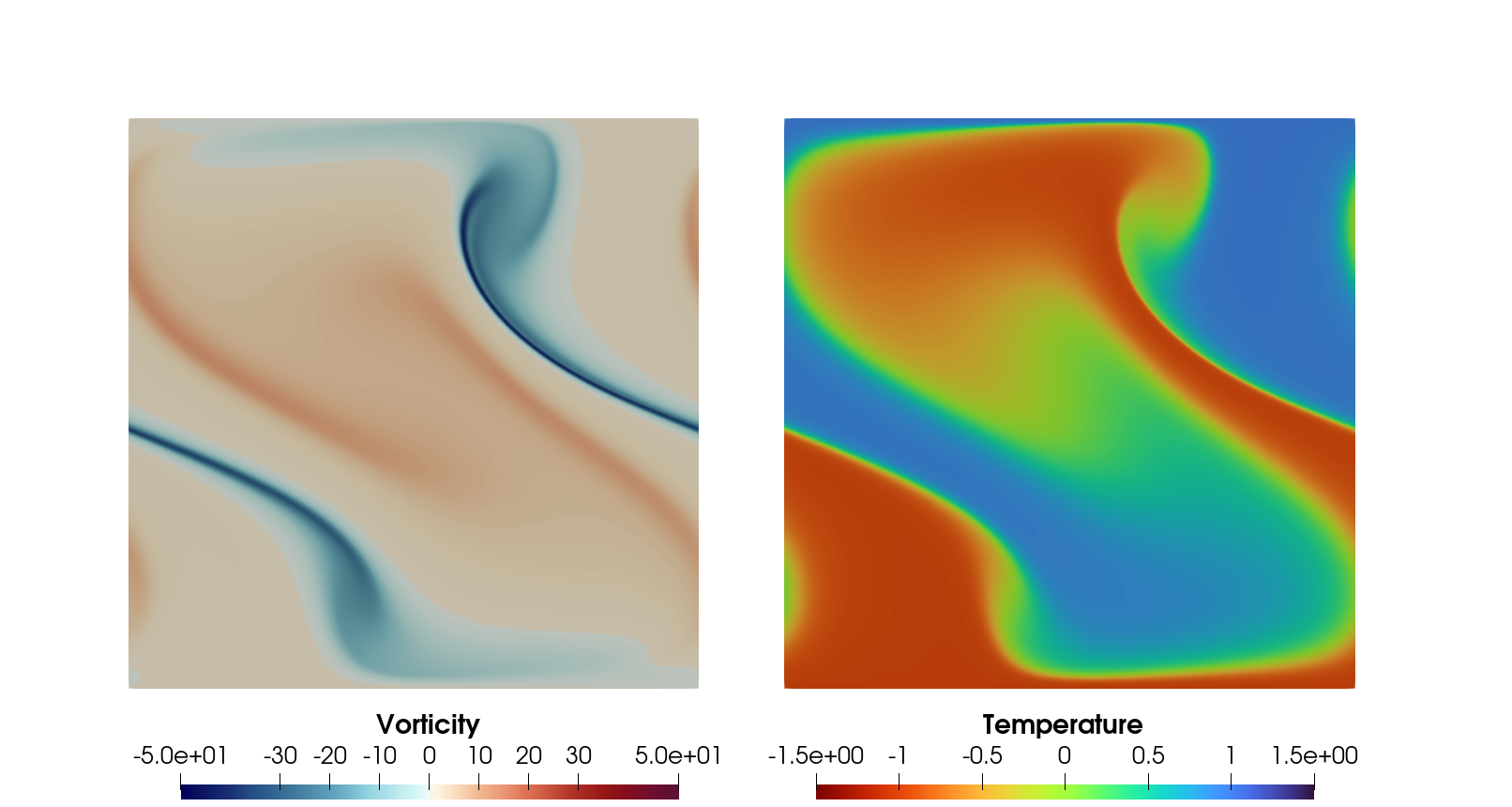}
\label{fig:example 1 - b}
\caption{}
 \end{subfigure}
\begin{subfigure}[b]{0.625\textwidth}
\centering
\includegraphics[width=\textwidth, trim=102 0 124 93, clip]{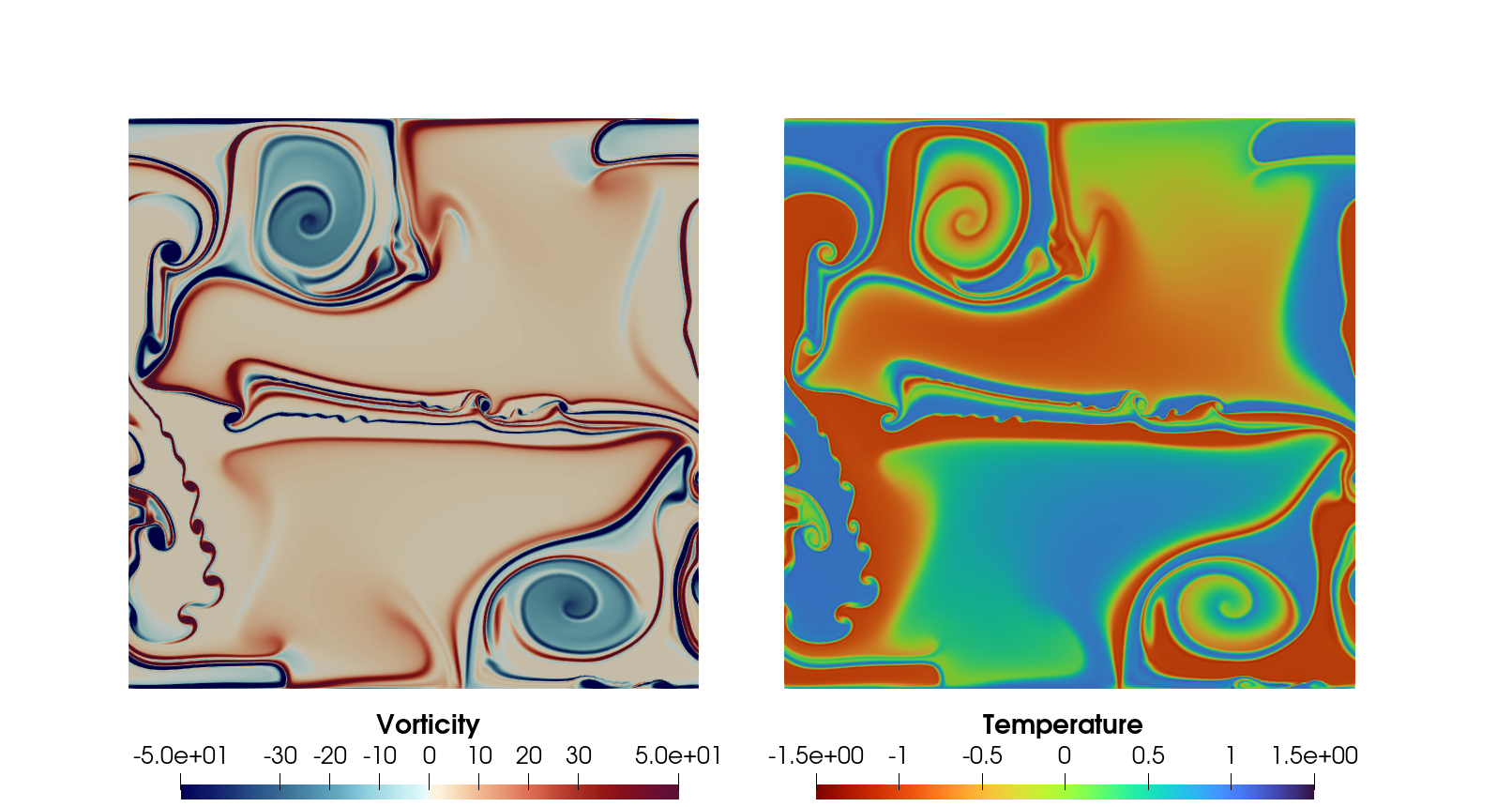}
\label{fig:example 1 - c}
\caption{}
 \end{subfigure}
\caption{Example numerics for \eqref{EBC-eqns-def2}. In this run, we chose the initial conditions for temperature and vorticity to be $\Theta(0, x, z) = \tanh(4(z - 0.5))$ and $\omega(0, x, z) = \sin(\pi x) \cos(z)$. The subfigures show snapshots of the flow at different stages of spin-up. }
\label{figure: example1}
\end{figure}

\bigskip\bigskip

\begin{figure}[h!]
\centering
\begin{subfigure}[b]{0.625\textwidth}
\centering
\includegraphics[width=\textwidth, trim=102 0 124 93, clip]{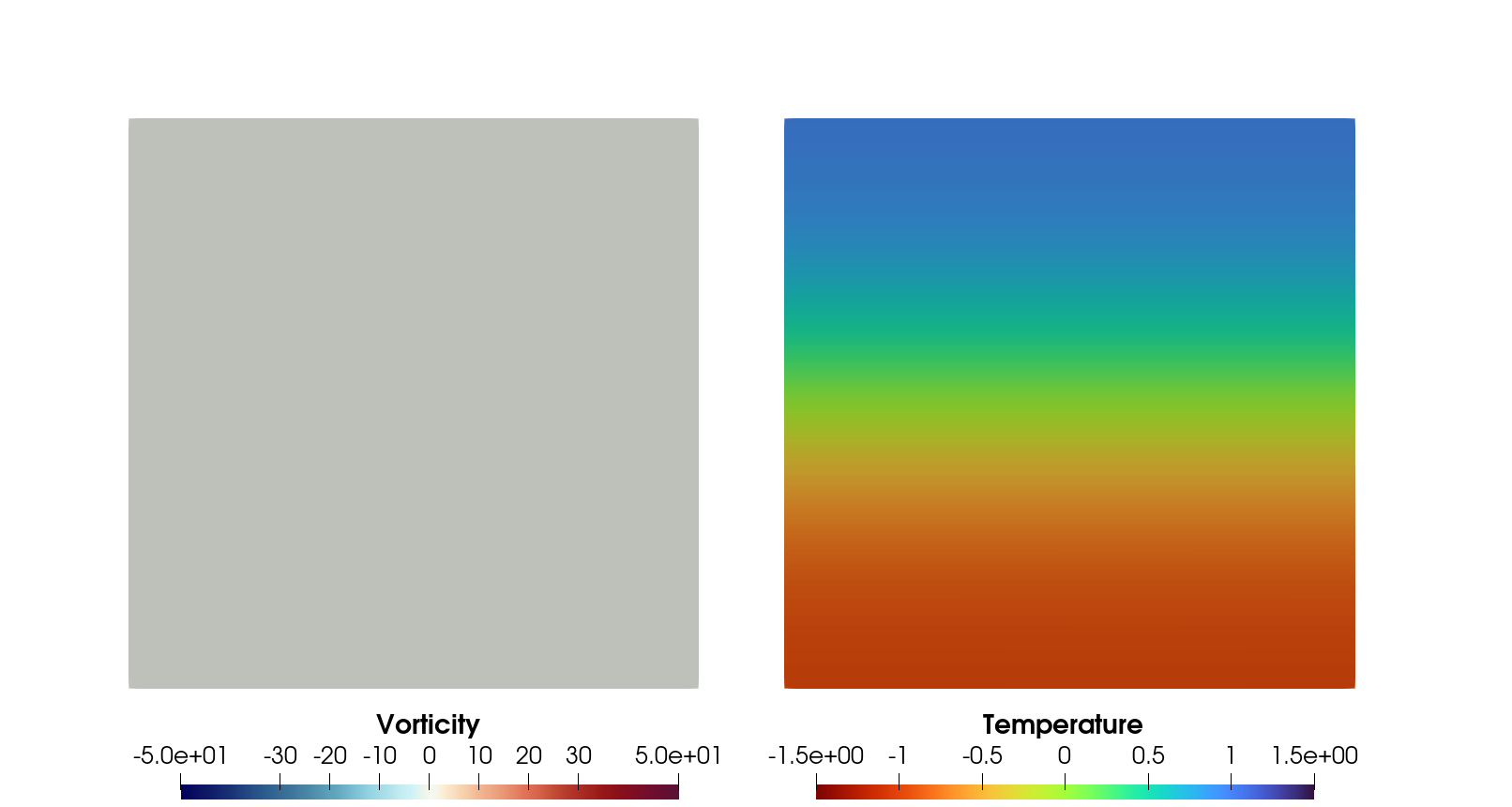}
\label{fig:example 2 - a}
\caption{}
\end{subfigure}
\begin{subfigure}[b]{0.625\textwidth}
\centering
\includegraphics[width=\textwidth, trim=102 0 124 93, clip]{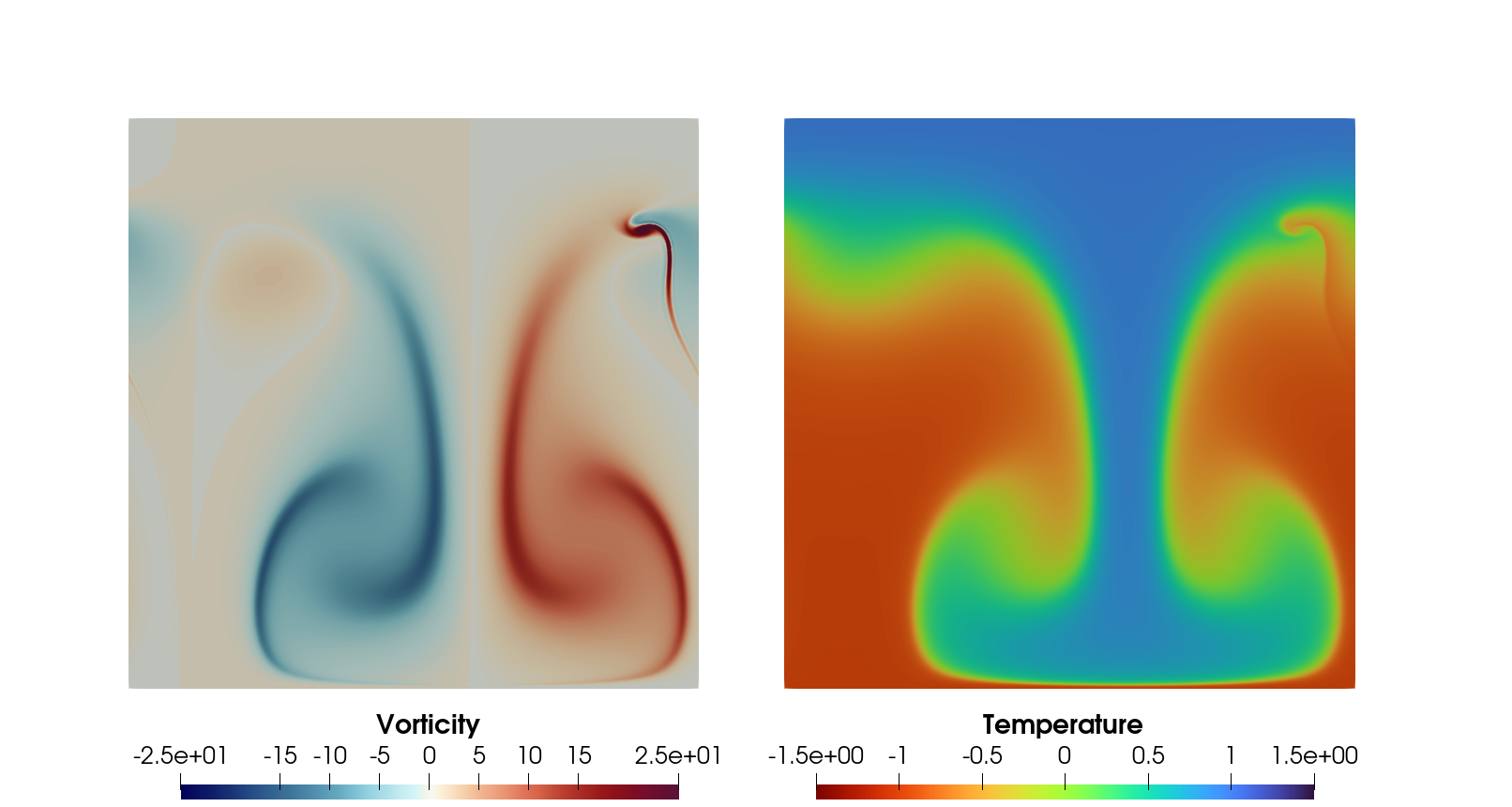}
\label{fig:example 2 - b}
\caption{}
 \end{subfigure}
\begin{subfigure}[b]{0.625\textwidth}
\centering
\includegraphics[width=\textwidth, trim=102 0 124 93, clip]{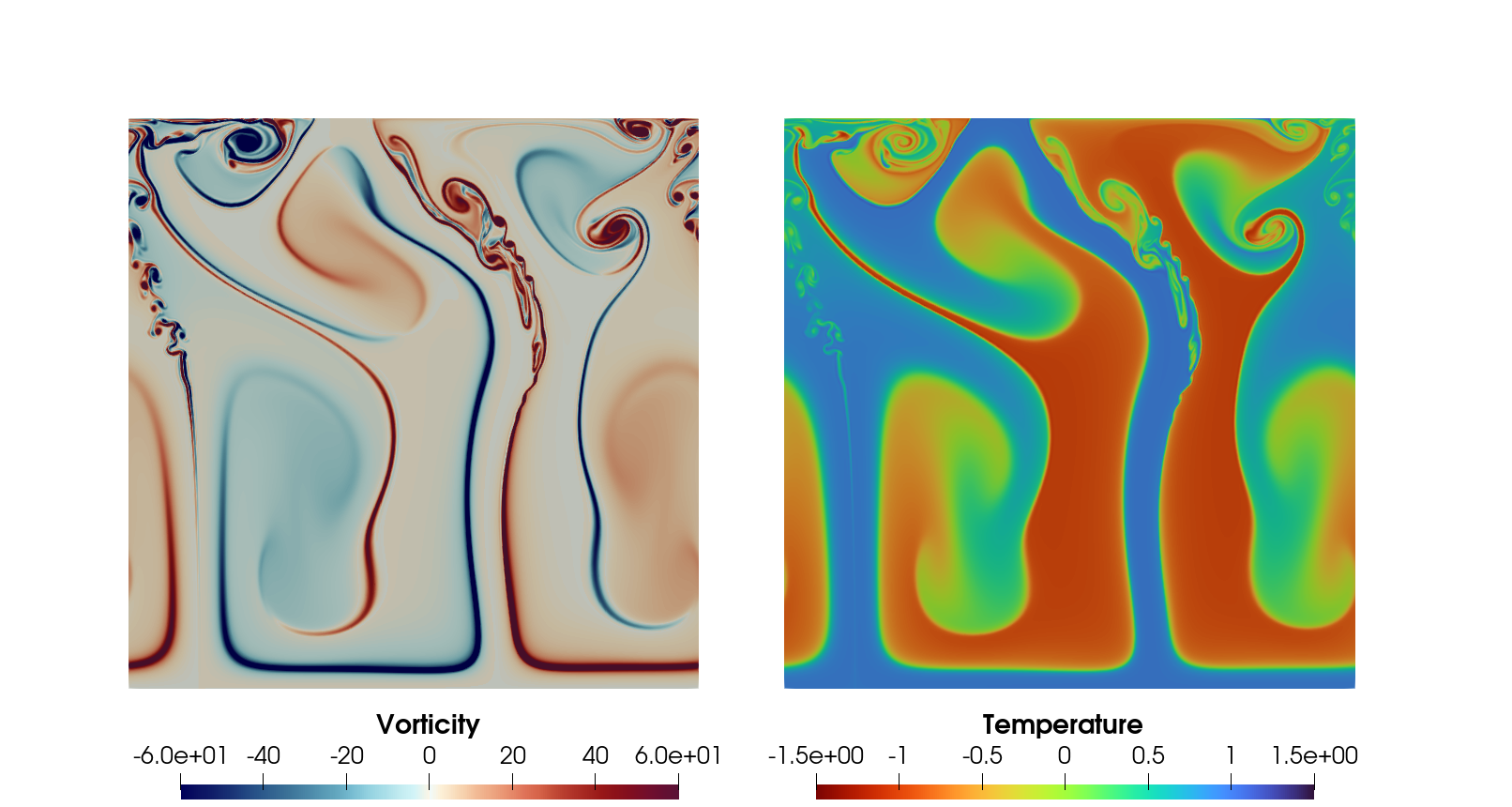}
\label{fig:example 2 - c}
\caption{}
 \end{subfigure}
\caption{Example numerics for \eqref{EBC-eqns-def2}. In this run, we chose the initial conditions for temperature and vorticity to be $\Theta(0, x, z) = \tanh(4(z - 0.5)) + 0.001\exp(-20((x-0.6)^2 + z^2)) $ and $\omega(0, x, z) = 0$. The subfigures show snapshots of the flow at different stages of spin-up. }
\label{figure: example2}
\end{figure}

\newpage

\begin{figure}[h!]
\centering
\begin{subfigure}[b]{0.625\textwidth}
\centering
\includegraphics[width=\textwidth, trim=102 0 124 93, clip]{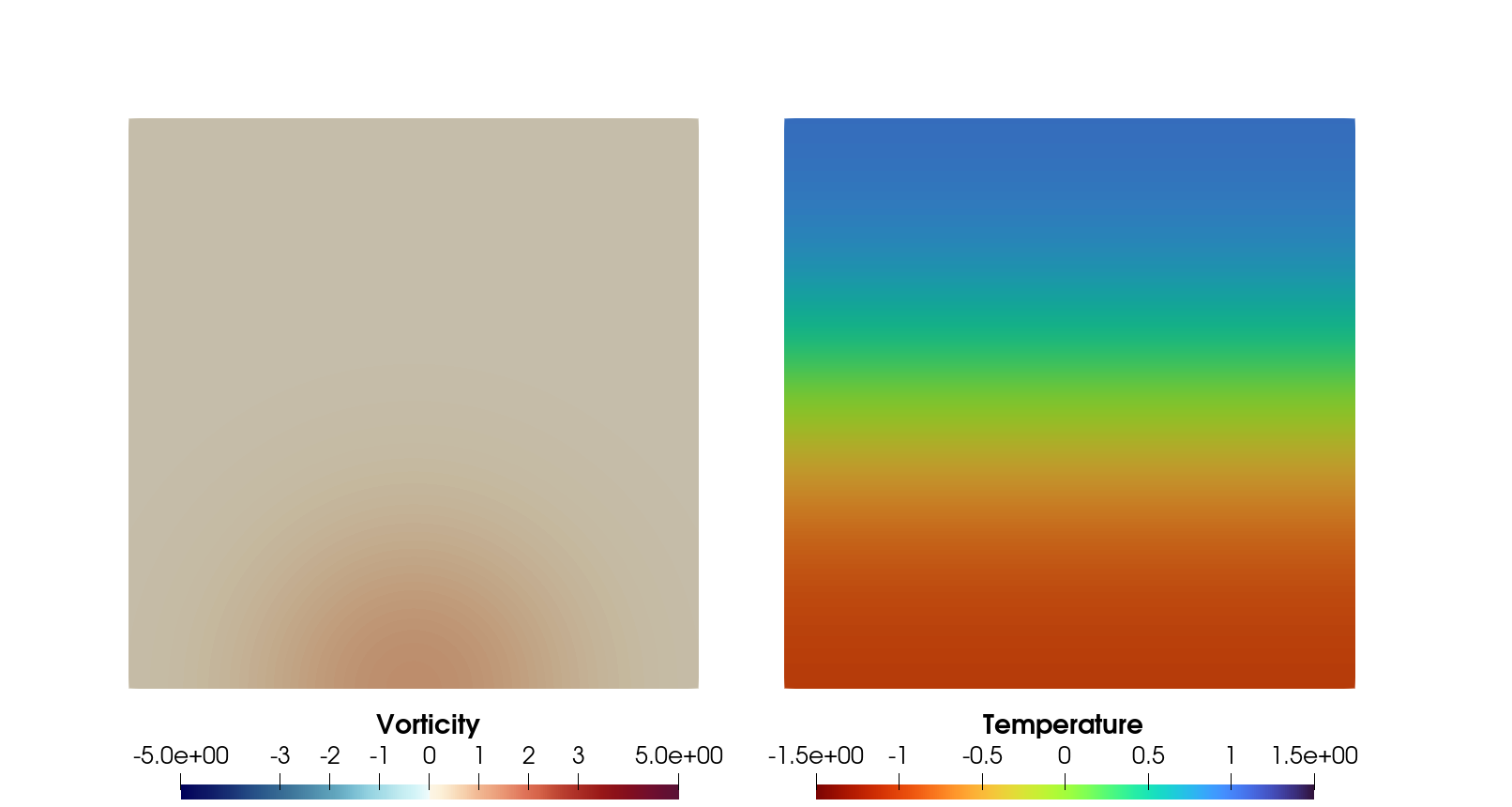}
\label{fig:example 3 - a}
\caption{}
\end{subfigure}
\begin{subfigure}[b]{0.625\textwidth}
\centering
\includegraphics[width=\textwidth, trim=102 0 124 93, clip]{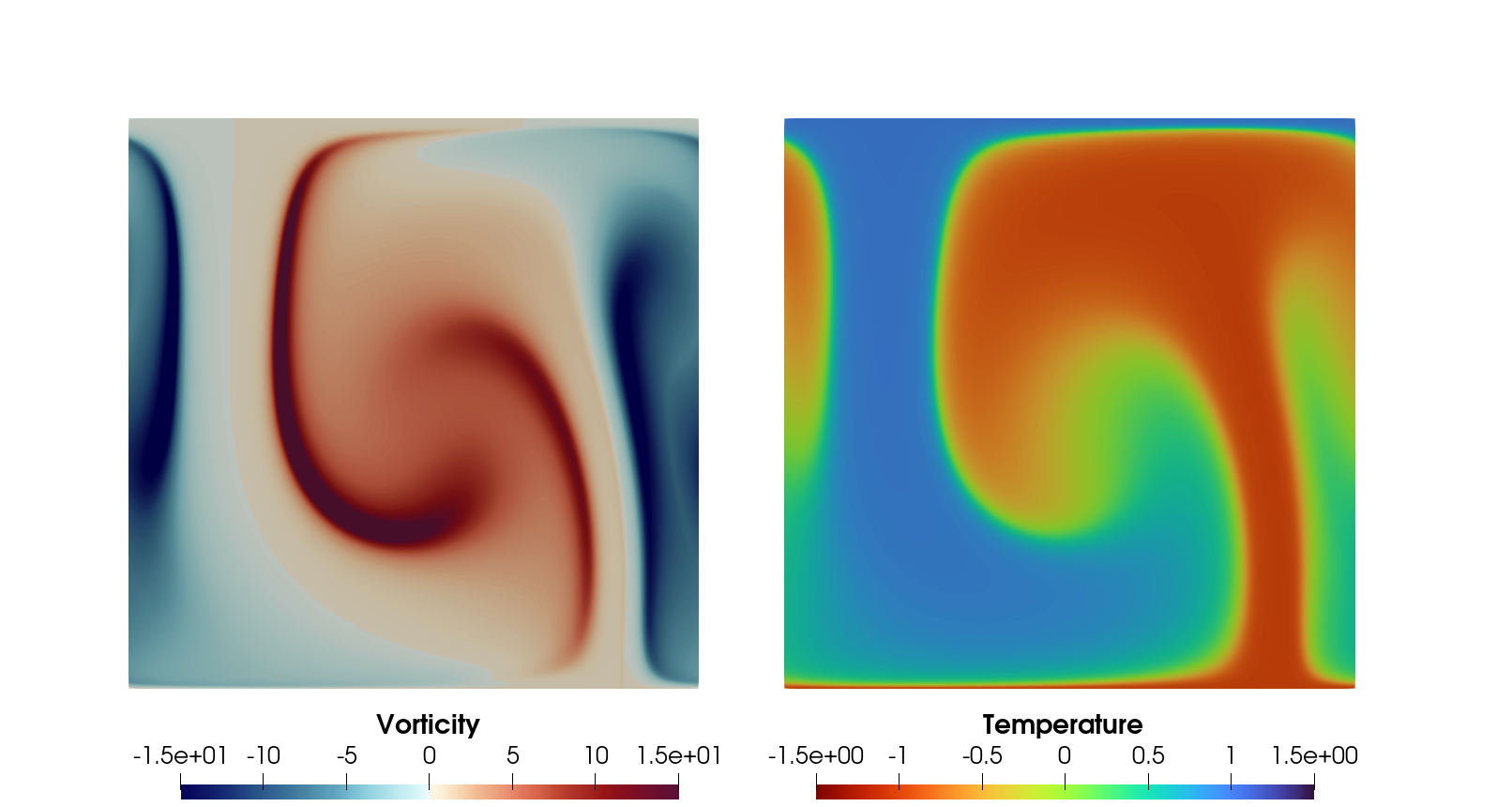}
\label{fig:example 3 - b}
\caption{}
 \end{subfigure}
\begin{subfigure}[b]{0.625\textwidth}
\centering
\includegraphics[width=\textwidth, trim=102 0 124 93, clip]{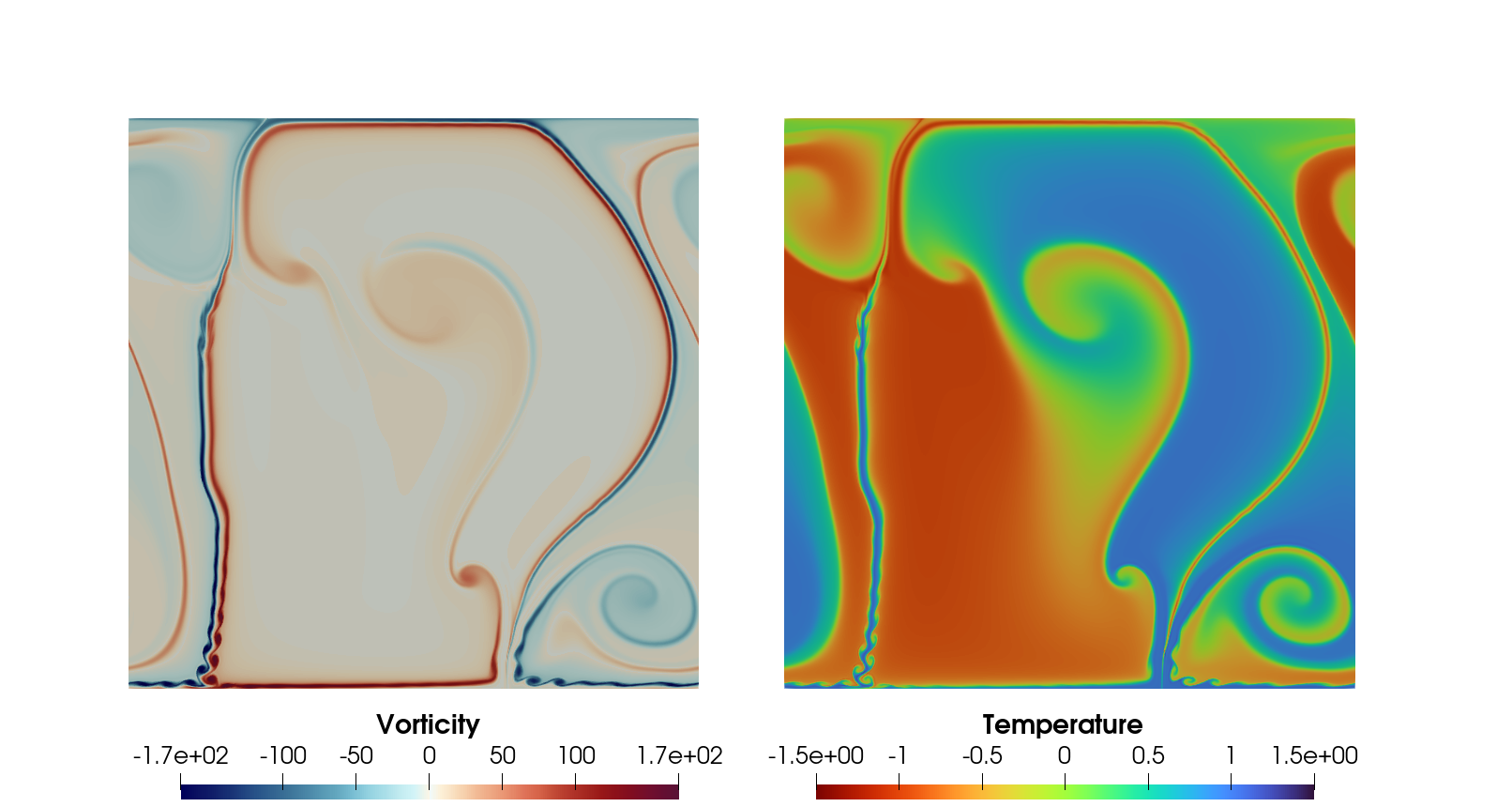}
\label{fig:example 3 - c}
\caption{}
 \end{subfigure}
\caption{Example numerics for \eqref{EBC-eqns-def2}. In this run, we chose the initial conditions for temperature and vorticity to be $\Theta(0, x, z) = \tanh(4(z - 0.5))$ and $\omega(0, x, z) = \exp[-10[(x-0.5)^2 + z^2]]$. The subfigures show snapshots of the flow at different stages of spin-up. }
\label{figure: example3}
\end{figure}

\newpage

\paragraph{\bf Conservation Laws for EBC}
The EBC equations in \eqref{EBC-eqns-def3} preserve the sum of the kinetic and thermal energies of the EBC system, made definite in sign by completing squares to find
\begin{align}
\begin{split}
H_{EBC}(\omega,\Theta) 
&= \int  \frac12\,  \omega \Delta^{-1}\omega - z \Theta \,dxdz
\\&= \frac12 \int  \omega \Delta^{-1}\omega + \big(\Theta -   z\big)^2  dxdz
- \frac12\int(\Theta^2 +  z^2) \,dxdz
\,.
\end{split}
\label{EBC-eqns-Ham}
\end{align}
The last (negative) term in the second line of \eqref{EBC-eqns-Ham} is immaterial, because $\int z^2dxdz$ is a non-dynamical constant and, as mentioned next, $\int{\Theta}^2dxdz$ is a Casimir in the Lie-Poisson Hamiltonian formulation of these EBC equations.

The deterministic EBC equations in \eqref{EBC-eqns-def3} also preserve an infinity of integral conservation laws, determined by two arbitrary differentiable functions of temperature $\Phi(\Theta)$ and $\Psi(\Theta)$ as
\begin{equation}
\mc{C}_{\Phi,\Psi}(\Theta,\omega) = \int_{\mathcal D}\Phi(\Theta ) + \omega \Psi(\Theta )\,\dd x\dd z\,.
\label{eq:casimirsEBC}
\end{equation}
That the EBC dynamics in \eqref{EBC-eqns-def1} preserves the energy in \eqref{EBC-eqns-Ham} and the family of integral quantities in \eqref{eq:casimirsEBC} can be verified by direct computations. However, the infinity of conservation laws in \eqref{eq:casimirsEBC} also indicates that the EBC system in \eqref{EBC-eqns-def1} possesses a rich mathematical structure which will help in diagnosing its solution behaviour. We investigate the geometrical aspects of this mathematical structure in the next few subsections.
\medskip


\medskip

\subsection*{Lie-Poisson Hamiltonian form of EBC}
The conservation laws in $\mc{C}_{\Phi,\Psi}$ signal a deeper mathematical content. For example, the EBC equations may be expressed in Lie-Poisson Hamiltonian form, in terms of a semidirect-product Lie-Poisson bracket given in matrix form as
\begin{equation}
 \frac{\partial }{\partial t} 
\begin{bmatrix}
\omega \\ \Theta
\end{bmatrix}
=
\begin{bmatrix}
J(\omega,\,\cdot\,) & J(\Theta,\,\cdot\,) 
\\ 
J(\Theta,\,\cdot\,)  & 0
\end{bmatrix}
\begin{bmatrix}
{\delta \mc H_{EBC}}/{\delta \omega} = \psi
\\ 
{\delta \mc H_{EBC}}/{\delta \Theta} = \Theta - z
\end{bmatrix}
=:
\Bigg\{\begin{bmatrix}
\omega \\ \Theta
\end{bmatrix} , \mc{H}_{EBC} \Bigg\}
\,.
\label{EBC-LPHam}
\end{equation}
The Lie-Poisson bracket relation can also be written in the shorter notation  
\begin{equation}
 \frac{\partial (\omega \,; \Theta)}{\partial t}  = \big\{(\omega \,;  \Theta)\,,\, \mc{H}_{EBC}\big\}. 
\label{short-note}
\end{equation}

\subsection*{Energy conservation of deterministic EBC flows}
The Hamiltonian $\mc{H}_{EBC}$ is conserved because the Lie-Poisson matrix operator defined in \eqref{EBC-LPHam} is skew symmetric in the $L^2$ pairing, so that $\frac{d}{dt}\mc{H}_{EBC} = \{\mc{H}_{EBC},\mc{H}_{EBC}\}=0$. The quantities $\mc{C}_{\Phi,\Psi}$ in \eqref{eq:casimirsEBC} are conserved because their variational derivatives are null eigenvectors of the Lie-Poisson matrix operator in \eqref{EBC-LPHam}. The type of conserved quantity such as  $\mc{C}_{\Phi,\Psi}$ that Poisson commutes with any Hamiltonian are called the \emph{Casimir functions} of the Lie-Poisson bracket \cite{HMR1998}.

\subsection*{Geometric considerations of deterministic EBC flows}

The EBC conservation laws in \eqref{eq:casimirsEBC} for divergence-free flow are in the same form as are satisfied in the dynamics of semidirect-product Lie-Poisson bracket for fluid flow in the variables $(\omega,\Theta)$ \cite{HMR1998}. Thus, the EBC model fits into the standard Hamiltonian framework for all ideal fluids with advected quantities, including their Casimir conservation laws, as explained further in Appendix \ref{appendix-B}.

\paragraph{\bf Summary of equilibrium EBC properties in Appendix \ref{appendix-B}.}
The planar EBC model in \eqref{EBC-eqns-def1} has been shown to be a Hamiltonian system for the semidirect-product Lie-Poisson bracket defined in equation \eqref{EBC-brkt3}. This result implies that the EBC system will possess all of the geometric properties belonging to the class of ideal fluid models with advected quantities. 

The geometric properties of this class of ideal fluid models are known \cite{HMR1998}. For our case, the geometric nature of the evolution of the EBC system written in its Hamiltonian form is discussed in the appendix \ref{appendix-A} where we identify its Lie-Poisson bracket with the coadjoint action of the semidirect-product Lie group of symplectic transformations defined in equation \eqref{EBC-brkt3}. 

Namely, the solutions of the EBC system in \eqref{EBC-eqns-def1} evolve by undergoing coadjoint motion along a level set of the Casimirs in equation \eqref{eq:casimirsEBC}. This coadjoint motion is governed by following a time-dependent path on the Lie group manifold of semidirect-product symplectic diffeomorphisms acting on the domain of flow $\mc{D}$. 

The association of the EBC system with the smooth flow of symplectomorphisms on  $\mc{D}$ may be a hidden factor in the success of the analysis of solution properties of EBC solutions. For example,  it was shown in  \cite{Foias-etal-RBCanalysis1987,Temam-RBCanalysis1997} that the 2D RBC system with viscosity and heat diffusivity has a finite-dimensional global attractor. 

The corresponding analytical properties for the stochastic counterpart of RBC may also follow suit for LA SALT EBC below, after we use the deterministic properties obtained here to introduce the Hamiltonian form of \emph{stochastic advection by Lie transport} (SALT) to model uncertainty in the dynamics of the EBC system discussed here. 

Having illuminated the  Lie-Poisson Hamiltonian structure of deterministic EBC dynamics,  next we will  begin discussing its stochastic counterpart. 


\section{Hamiltonian formulation of SALT for EBC}\label{sec: SALT-EBC}
\label{sec:HamiltonSEBC}


By following \cite{Holm2015}, the Hamiltonian formulation of SALT for EBC can be derived by augmenting its deterministic Hamiltonian in \eqref{EBC-eqns-Ham} to add a stochastic Hamiltonian whose flow under the Lie-Poisson bracket \eqref{EBC-brkt3} induces a Stratonovich stochastic flow along the characteristics of the symplectic vector field generated by the stochastic Hamiltonian process $\langle \omega \,,\,{\chi}(x,z) \rangle \circ dW_t$. The semimartingale total Hamiltonian is then given by 
\begin{align}
{\rm d}h_{EBC} = H_{EBC}(\omega,\Theta)\dd t + \int_{\mc D} \omega \,{\chi}(x,z) \dd x \dd z\circ \dd W_t
\,,\label{SALT-HamShift-EBC}
\end{align}
where the deterministic part $H_{EBC}$ of the Hamiltonian ${\rm d}h_{EBC}$ is given in equation \eqref{EBC-eqns-Ham}.
Next, the Lie-algebra valued variations with respect to the \emph{dual} Lie-algebra valued variables $(\omega,\Theta)\in (f_1\circledS f_2)^*$ for EBC are expressed in semimartingale form as
\begin{align}
\delta\big({\rm d}h_{EBC}\big) = \int \Big(\psi {\rm d}t + {\chi}(x,z) \circ \dd W_t \Big)\delta \omega
+ \big( \Theta - z\big)\dt \,\delta \Theta
\, \dd x \dd z
\,.\label{SALT-deltaHam-EBC}
\end{align}
SALT EBC dynamics is then expressed geometrically using the coadjoint operator ${\rm ad}^*: \mathfrak{g}\times \mathfrak{g}^*\to \mathfrak{g}^*$  in \eqref{SDP-EBC-LP-ad-star} in stochastic integral form as
\begin{align}
\big({\rm d} \omega ; {\rm d}\Theta\big) 
= - \,{\rm ad}^*_{\big(\psi {\rm d}t + {\chi}(x,y) \circ \dd W_t \,;\, (\Theta - z)dt \big)}
\big(\omega ; \Theta\big)
\,.
\label{SDP-EBC-LP-SALT}
\end{align}
Thus, SALT EBC solutions evolve by stochastic coadjoint motion under right semidirect-product action of symplectic diffeomorphisms on the dual of its Lie algebra with dual (momentum-map) variables $(\omega ; \Theta)\in (f_1\circledS f_2)^*$. The symplectic semimartingale vector fields
$({\delta ({\rm d}h_{EBC})}/{\delta \omega}\,;\,{\delta ({\rm d}h_{EBC})}/{\delta \Theta})\in f_1\circledS f_2$ whose characteristic curves generate the stochastic coadjoint solution behaviour are obtained from the variations in equations \eqref{SALT-deltaHam-EBC}.

One may write the SALT EBC equations in Hamiltonian  matrix operator form as
\begin{equation}
 \frac{\partial }{\partial t} 
\begin{bmatrix}
\omega \\ \Theta
\end{bmatrix}
=
\begin{bmatrix}
J(\omega,\,\cdot\,) & J(\Theta,\,\cdot\,) 
\\ 
J(\Theta,\,\cdot\,)  & 0
\end{bmatrix}
\begin{bmatrix}
{\delta \mc ({\rm d}h_{EBC})}/{\delta \omega} = \psi {\rm d}t + {\chi}(x,z) \circ \dd W_t
\\ 
{\delta \mc ({\rm d}h_{EBC})}/{\delta \Theta} = (\Theta - z)\dt
\end{bmatrix}
\,.
\label{EBC-LPHam-SALT}
\end{equation}

In standard form, the SALT version of the EBC equations \eqref{EBC-eqns-def1} for convection in a planar vertical slice becomes
\begin{align}
\begin{split}
{\rm d} \omega + (\bu\, \dt + \bm{\xi}_i \circ \dd W^i)\cdot \nabla \omega
&= \Theta_x\,\dt
\,,
\\
{\rm d} \Theta + (\bu\, \dt + \bm{\xi}_i \circ \dd W^i)\cdot \nabla \Theta &= 0
\,,
\end{split}
\label{EBC-eqns-SALT}
\end{align}
where 
\[
\bu :=  \nabla^\perp\psi = (-\,\psi_z\,,\, \psi_x)\,,
\quad\hbox{and}\quad
\bm{\xi}_i := \nabla^\perp{\chi}_i = (-\,{\chi}_{i,z}\,,\, {\chi}_{i,x})\,,
\]
The boundary conditions are the same as before and the noise terms should be compatible with those boundary conditions .

Here, the fixed symplectic vector fields $\bs{\xi}_i(x,z):=\nabla^\perp{\chi}_i(x,z)$ need to be obtained from observed or simulated data, using the chain of procedures of calibration, analysis and assimilation of fluid data established in \cite{cotter2018modelling,cotter2019numerically}. 

\paragraph{The SALT algorithm.}  The multiscale features of the SALT algorithm can be illustrated as a series of three steps 
  applied on different levels of integration, projection and comparison.\bigskip

\begin{tabular*}{1.0\textwidth}%
     {@{\extracolsep{1cm}}l l}
\begin{minipage}[ ]{1.25in}
\begin{picture}(1,1.4)(0,50)
\includegraphics[width=3.75cm]{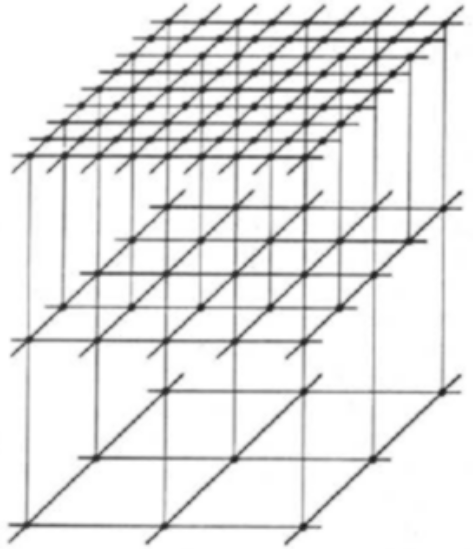} 
\put(-108,-13){\bf Multiscale comparisons.}
\end{picture}
\end{minipage}
& 
  \begin{minipage}[  ]{5.5in}\normalsize$\,$\vspace{-4mm}

(1) Fine-scale `Truth'  ($1024^2$) obtained, e.g., from fluid PDE simulations.\vspace{1cm}

(2) Coarse-scale SALT SPDE ($64^2$) 
for the stochastic 2D EPC 
Lagrangian advection velocity.
This step quantifies the uncertainty of the coarse simulation.\vspace{8mm}

(3) Particle filtering using selected data $(4^2)$
This step reduces uncertainty.
\vspace{1cm}
\end{minipage}
\end{tabular*}

Briefly put, the multiscale SALT procedure using data analysis and assimilation executes the following steps beginning from fine-scale computational simulations:

\begin{itemize}
\item 
Simulate the Eulerian PDE on a fine scale and obtain their correlation statistics in order to calibrate the SPDE on a coarser scale.

\item 
Create the data base from, e.g., Lagrangian tracers from fine-scale deterministic numerical simulations of an Eulerian PDE system. 

\item Calibrate the Eulerian correlation eigenvectors $\bs{\xi}_{i}(x,z)$ as emperical orthogonal functions derived from either  the fine-scale PDE numerics, or from the observed data. 

\item Derive the SALT SPDE from an approximate Drift + Noise velocity decomposition based on homogenisation in time, as discussed in \cite{CotterGottwaldHolm2017}. 

\item Simulate the SALT SPDE on the coarse scale, to quantify uncertainty arising from combined modelling error \& numerical simulation error.

\item Assimilate a selection of observed data at high resolution by using particle filtering and thereby reduce the uncertainty. 

\end{itemize}

For an explicit example of the SALT data calibration and assimilation procedure for the well-known quasigeostrophic fluid equations, see \cite{cotter2019numerically}.

\section{Two variants of the SALT approach: SFLT and LA SALT} \label{sec:SEBC-2variants}

Let us make some brief remarks about two potentially interesting variants of the SALT approach. The first variant provides an alternative approach to developing stochastic fluid equations for EBC that enables stochastic perturbations that preserve energy. The second variant allows one to define the `climate change' of the EBC equations by deriving equations for their expectation dynamics and the evolution of their statistical properties. These two variants, called Stochastic Forcing by Lie Transport (SFLT) and Lagrangian Averaged SALT (LA SALT), respectively, will be discussed briefly below as open problems. 

\subsubsection*{Stochastic Forcing by Lie Transport (SFLT)}
An energy-preserving variant of SALT can be formulated for EBC by introducing semimartingale terms into the Lie-Poisson bracket, as follows \cite{HH2021}
\begin{align}
\begin{split}
 {\rm d} 
\begin{bmatrix}
\omega\\ \Theta
\end{bmatrix}
&=
\begin{bmatrix}
J(\omega \dt + \xi(x,z) \circ \dd W_t ,\,\cdot\,) & J(\Theta \dt + \xi(x,z) \circ \dd W_t ,\,\cdot\,) 
\\ 
J(\Theta \dt + \xi(x,z) \circ \dd W_t ,\,\cdot\,)  & 0
\end{bmatrix}
\begin{bmatrix}
{\delta {\mc H}_{EBC} } /{\delta \omega} = \psi 
\\ 
{\delta ({\mc H}_{EBC})}/{\delta \Theta} = (\Theta - z)
\end{bmatrix}
\\&= - 
\begin{bmatrix}
(\bu\cdot \nabla)\big(\omega \dt + \xi(x,z) \circ \dd W_t \big) 
- \nabla^\perp(\Theta - z)\cdot \nabla\big(\Theta \dt + \xi(x,z) \circ \dd W_t \big) 
\\ 
(\bu\cdot \nabla) \big(\Theta \dt + \xi(x,z) \circ \dd W_t \big)  
\end{bmatrix}
\,,
\end{split}
\label{EBC-SFLTeqn}
\end{align}
with velocity $\bu =\nabla^\perp \psi$. Because the Hamiltonian matrix operator for the  Lie-Poisson bracket is skew-symmetic under the $L^2$ pairing, we now have ${\rm d}{\mc H}_{EBC}=0$.
Thus, the SFLT formulation of EBC preserves the deterministic energy Hamiltonian, ${\mc H}_{EBC}$, even though its dynamics is stochastic. In contrast, the SALT approach preserves the Casimirs, but not the energy. The application of SFLT in EBC is an open problem, which will be deferred until later. 
\\

\subsubsection*{The EBC `climate' according to Lagrangian Averaged SALT (LA SALT)}

One may regard the SALT EBC equations as a minimal model for climate science, by creating the Lagrangian Averaged SALT (LA SALT)  EBC equations. This is done by modifying the SALT EBC equations in \eqref{EBC-LPHam-SALT} to replace the stream function $(\psi)$ for the drift velocity by its expectation $\expect(\psi)$, in the Hamiltonian formulation, following \cite{Alonso-Oran-etal-JStaPhys}, based on a method introduced in \cite{McKean1966}. 

\begin{equation}
\dd
\begin{bmatrix}
\omega \\ \Theta
\end{bmatrix}
=
\begin{bmatrix}
J(\omega,\,\cdot\,) & J(\Theta,\,\cdot\,) 
\\ 
J(\Theta,\,\cdot\,)  & 0
\end{bmatrix}
\begin{bmatrix}
\expect(\psi) {\rm d}t + {\chi}(x,z) \circ \dd W_t
\\ 
(\Theta - z)\dt
\end{bmatrix}
\,.
\label{EBC-LPHam-LASALT}
\end{equation}
Although the LA SALT EBC system in \eqref{EBC-LPHam-LASALT} is no longer Hamiltonian, it still preserves the same Casimir functionals as for SALT EBC, because it possesses the same Lie-Poisson bracket structure as for the SALT EBC system. In this section, we will write out the dynamical equations for expectations $(\expect(\omega),\expect(\Theta))$ and fluctuations $(\omega',\Theta')$, as well as the evolution equations for the statistics of the LA SALT  EBC system in \eqref{EBC-LPHam-LASALT}. Our treatment will follow the work for the Boussinesq case in \cite{Alonso-Oran-etal-JStaPhys}. This section also links well with the SALT and LA SALT treatments of the Lorenz-63 model of EBC in \cite{GHL2019}. 

In standard form, the LA SALT version of the EBC equations \eqref{EBC-eqns-def1} for convection in a planar vertical slice is given by
\begin{align}
\begin{split}
{\rm d} \omega + \big(\expect(\bu)\, \dt + \bm{\xi}_i \circ \dd W^i \big)\cdot \nabla \omega
&= \Theta_x\,\dt
\,,
\\
{\rm d} \Theta + \big(\expect(\bu)\, \dt + \bm{\xi}_i \circ \dd W^i \big)\cdot \nabla \Theta &= 0
\,,
\end{split}
\label{EBC-eqns-LASALT}
\end{align}
where 
\[
\expect(\bu) :=  \nabla^\perp\expect(\psi) = (-\,\expect(\psi)_z\,,\, \expect(\psi)_x)\,,
\quad\hbox{and}\quad
\bm{\xi}_i := \nabla^\perp{\chi}_i = (-\,{\chi}_{i,z}\,,\, {\chi}_{i,x})\,,
\]
The boundary conditions are the same as before and the noise terms should be compatible with those boundary conditions.

Upon passing to the It\^o formulation, the LA SALT EBC equations \eqref{EBC-eqns-LASALT} transform into 
\begin{align}
\begin{split}
{\rm d} \omega + \expect(\bu)\cdot \nabla \omega\, \dt 
+ \bm{\xi}_i  \cdot \nabla \omega \dd W^i 
&= \Theta_x\,\dt +  \frac12 \bm{\xi}_i\cdot\nabla\big(\bm{\xi}_i\cdot\nabla \omega\big) \dt
\,,
\\
{\rm d} \Theta + \expect(\bu)\cdot \nabla \Theta\, \dt 
+ \bm{\xi}_i   \cdot \nabla \Theta \dd W^i &=  \frac12 \bm{\xi}_i\cdot\nabla\big(\bm{\xi}_i\cdot\nabla \Theta\big) \dt
\,.
\end{split}
\label{EBC-eqns-LASALT-Ito}
\end{align}
Then taking the expectation of the It\^o form yields a closed system for the expected solutions, as
\begin{align}
\begin{split}
\frac{\p}{\p t} \expect(\omega) + \expect(\bu)\cdot \nabla \expect(\omega) 
&= \expect(\Theta)_x +  \frac12 \bm{\xi}_i\cdot\nabla\big(\bm{\xi}_i\cdot\nabla \expect(\omega)\big)
\,,
\\
\frac{\p}{\p t}\expect(\Theta) + \expect(\bu)\cdot \nabla \expect(\Theta)
&=  \frac12 \bm{\xi}_i\cdot\nabla\big(\bm{\xi}_i\cdot\nabla \expect(\Theta)\big)
\,.
\end{split}
\label{EBC-eqns-LASALT-expect}
\end{align}
Define the fluctuations as $\omega' = \omega - \expect(\omega)$ and $\Theta' = \Theta - \expect(\Theta)$. Taking the corresponding differences in the two previous equation sets then yields the following set of linear fluctuation equations,
\begin{align}
\begin{split}
\dd\omega' + \expect(\bu)\cdot \nabla \omega' \dt 
+ {\color{black} \bm{\xi}_i \cdot \nabla \omega \dd W^i} 
&= \Theta'_x \dt +  \frac12 \bm{\xi}_i\cdot\nabla\big(\bm{\xi}_i\cdot\nabla \omega'\big) \dt
\,,
\\
\dd \Theta' + \expect(\bu)\cdot \nabla\Theta'\dt +
{\color{black}  \bm{\xi}_i \cdot \nabla \Theta \dd W^i} 
&=  \frac12 \bm{\xi}_i\cdot\nabla\big(\bm{\xi}_i\cdot\nabla \Theta'\big)\dt
\,.
\end{split}
\label{EBC-eqns-LASALT-fluct}
\end{align}
One realises that the equations \eqref{EBC-eqns-LASALT-expect} and \eqref{EBC-eqns-LASALT-fluct} provide the \emph{entire history} of the solutions for the expectations $\expect(\omega)$ and $\expect(\Theta)$ throughout the domain of flow. Once the expectation equations \eqref{EBC-eqns-LASALT-expect} have been solved, the equations for the instantaneous stochastic variables \eqref{EBC-eqns-LASALT-fluct} become \emph{linear} It\^o stochastic transport equations, which are driven by the solutions of equations \eqref{EBC-eqns-LASALT-expect}, whose entire history is obtained separately. Then one realises that the coupled system \eqref{EBC-eqns-LASALT-expect} and \eqref{EBC-eqns-LASALT-fluct} is closed, since the variables $(\omega,\Theta)$ and their corresponding variational derivatives are related linearly. These linear relations arise because the deterministic part $H_{EBC}$ of the Hamiltonian ${\rm d}h_{EBC}$ in \eqref{SALT-HamShift-EBC} is quadratic; so, its variations are linear.  

The fluctuation equations, in turn, imply evolutionary equations for the statistical moments of the LA SALT EBC equations; for example, the evolutionary equations for the covariances $\expect(\omega'\omega') = \omega^{(2)}$ and $\expect(\Theta'\Theta')=\Theta^{(2)}$. 

\begin{equation}
\left\{
\begin{array}{rl}
    \partial_{t} \omega^{(2)} + \E{\bu}\cdot \nabla \omega^{(2)}  
    &= 
    \displaystyle\sum_k \left(\frac12 (\bm{\xi}_k \cdot \nabla)^2 \omega^{(2)} + \Big((\bm{\xi}_k\cdot \nabla) \E{\omega}\Big)^2\right) 
    + { {\bf 2}} \E{\omega' \Theta'_x },  \\
    \partial_{t}\Theta^{(2)}+ \E{\bu}\cdot \nabla \Theta^{(2)} 
    &= 
    \displaystyle\sum_k \left(\frac12 (\bm{\xi}_k \cdot \nabla)^2 \Theta^{(2)} +  \Big(\big(\bm{\xi}_k\cdot \nabla) \E{\Theta}\Big)^2\right) ,
\end{array}
\right.
\label{omega-theta-covar-eqs}
\end{equation}
where $\omega^{(2)} := \E{(\omega')^2}$ and $\Theta^{(2)} := \E{(\Theta')^2}$. To close system \eqref{omega-theta-covar-eqs} one requires an evolution equation for the expected term $\E{\omega'    \Theta'_x }$ in the $\omega^{(2)}$ equation. The required evolution equation can be sought by taking the partial time derivative of $\E{\omega'   \Theta'_x }$, applying the stochastic product rule and substituting from the fluctuation evolution equations in \eqref{EBC-eqns-LASALT-fluct}. This calculation eventually yields the following equation, 
\begin{align}
\begin{split}
    &\partial_t \E{\omega'\Theta'_x} + \E{\bu}\cdot \nabla \E{\omega'\Theta'_x}  \\
    &= 
    \frac12 (\bm{\xi}_i \cdot \nabla)^2  \E{\omega'\Theta'_x} + (\bm{\xi}_i\cdot\nabla\E{\omega})( \bm{\xi}_i\cdot\nabla\E{\Theta_x})
    + \E{\Xi} 
    + \E{(\Theta'_x)^2} 
\,,\end{split}
\label{EBC-eqns-LASALT-moment-dyn}
\end{align}
where the quantity $\Xi$ in the moment dynamics \eqref{EBC-eqns-LASALT-moment-dyn} is defined to be,
\begin{equation}
    \Xi = - \omega'\partial_x\E{\bu}\cdot \nabla\Theta'
    +  \frac12\omega'
    \left(
        \partial_x(\bm{\xi}_i)\cdot\nabla\big(\bm{\xi}_i\cdot\nabla \Theta'\big)
        +  
        \bm{\xi}_i\cdot\nabla\big(\partial_x\bm{\xi}_i\cdot\nabla \Theta'\big)
    \right)
    +  (\partial_x\bm{\xi}_i \cdot \nabla \Theta)(\bm{\xi}_i \cdot \nabla \omega) 
    \,.
    \label{Xi-def}
\end{equation}

The last term  $\E{(\Theta'_x)^2}=(\Theta_x')^{(2)}$ in \eqref{EBC-eqns-LASALT-moment-dyn} is the covariance of the temperature gradient $\Theta'_x$, which drives the vorticity dynamics in \eqref{EBC-eqns-LASALT}. The necessary evolution equation for $(\Theta_x')^{(2)}$ can be derived by applying $\p_x$ to the second equation in \eqref{EBC-eqns-LASALT-fluct} and recalculating as before, leading to
\begin{align}
\begin{split}
    &
    \partial_t\E{(\Theta'_x)^2}\\
    &= 
    -2\E{\Theta'_x\partial_x(\E{\bu}\cdot\nabla \Theta')} 
    + \E{\Theta'_x \partial_x (\bm{\xi}_i \cdot \nabla (\bm{\xi}_i \cdot \nabla \Theta'))}
    + \E{(\partial_x (\bm{\xi}_i \cdot \nabla \Theta))^2} 
\,,\end{split}
\end{align}
which includes covariances of higher order gradients of $\Theta$.
Thus, the evolution equations for the covariances of gradients of LA SALT flow variables do not close. However, if the LA SALT evolution equation for $\Theta$ in \eqref{EBC-eqns-LASALT-Ito} is wellposed, then one may glean some information about the evolution of the statistical properties of the LA SALT solutions by simulating different $\Theta_x$ solutions and estimating their moments.
\smallskip

We have the following results for the LA SALT EBC equations: 

(i) The expectations of the solutions $\expect(\omega)$ and $\expect(\Theta)$ satisfy a closed nonlinear chaotic system that is isomorphic to the deterministic EBC equations with dissipation in equation \eqref{DOBeqns}. The initial value equation for these equations has a global strong solution, \cite{WenCharlie-etal2020}.

(ii) The stochastic evolution of the fluctuations in \eqref{EBC-eqns-LASALT-fluct} is linear.

(iii) The evolution equations for the covariances of the fluctuations $\expect(\omega'\omega')$, $\expect(\Theta'\Theta')$, $\E{\omega'  \Theta'_x }$,  and $\expect(\Theta_x'\Theta_x')$ form a system of partial differential equations (PDE) which is not closed.  \smallskip

Computational simulations of the PDE \eqref{omega-theta-covar-eqs} and \eqref{EBC-eqns-LASALT-moment-dyn} for the covariance dynamics of the LA SALT EBC fluctuations will be left as another open problem which is to be deferred until later, after the data calibration and assimilation steps in the SALT procedure have been accomplished. As a finite dimensional start, though, the investigation of the LA SALT expectation and fluctuations for the Lorenz-63 model of RBC have already been initiated in \cite{GHL2019}.
\bigskip

\section{Other open problems}\label{sec: Conclusion}

The previous sections have culminated by introducing the SALT, SFLT and LA SALT versions of the SEBC equations for convection in a vertical plane.  The deterministic EBC instability occupies a prominent position in the literature of mathematical analysis -- particularly in the contributions of CR Doering. 

Hopefully, the introduction of stochastic Euler-Boussinesq convection (SEBC) will stimulate more progress for stochastic analysis of convection. Much in the analysis of these equations has already been accomplished. 
The local existence, uniqueness and well-posedness of weak solutions for SEBC as well as  blow-up criteria in the sense of \cite{BKM1984} have been proved  in \cite{Alonso-Oran-bDLeon-JNLS2020}. 
Proof of the global well-posedness of strong solutions of the LA SALT version of the SEBC equations has been accomplished in \cite{Alonso-Oran-etal-JStaPhys}. 
However, as we have emphasised, the analysis of the SFLT version of SEBC remains as an open problem. 

Much more also remains to be done in the application of SEBC in data assimilation. For example, one could follow the plan of theoretical and computational analysis of the SALT procedure, as developed for the quasi-geostrophic model in \cite{CHLMP2021}. This is work in progress.

Finally, we mention that (i) the SALT approach discussed here has recently been generalised to the theory of rough paths in \cite{CHLN2022}; and (ii) the LA SALT approach has been applied and analysed in classic climate models of ocean-atmosphere introduction in 
\cite{CHK-2022}.

\subsection*{Acknowledgements}
This paper was written in appreciation of the late Charlie Doering's enthusiastic contributions to the mathematical analysis of fluid convection dynamics over the many years of his marvelous career.
We are grateful to our friends, colleagues and collaborators for their advice and encouragement in these matters. DH is grateful to P. Constantin, J. D. Gibbon, C. D. Levermore, E. S. Titi and our late friend C. Foias for shared moments of inspiration and camaraderie with Charlie Doering. Both DH and WP thank C. J. Cotter, D. Crisan, I. Shevchenko and our colleagues in the STUOD project for collaborations in developing the SALT algorithm. We also thank R. Hu and S. Patching for developing the SFLT algorithm, and we thank D. Crisan, T. D. Drivas and J.-M. Leahy for help with developing the LA SALT algorithm. DH and WP have been partially supported during the present work by European Research Council (ERC) Synergy grant STUOD - DLV-856408.


\appendix 
\section{Lie-Poisson Hamiltonian structure of deterministic EBC}\label{appendix-A}

\subsection{Geometric properties of the Lie-Poisson bracket for EBC dynamics}

\begin{proposition}[Jacobi identity]\label{LPB-EBC5}
The Lie-Poisson bracket in \eqref{short-note} satisfies the Jacobi identity.  
\end{proposition}

\begin{proof}
This proposition could be demonstrated by direct computation using the well-known properties of the Jacobian of function pairs. The proof given here, though, will illustrate the geometric properties of the EBC system in \eqref{EBC-eqns-def3} and thereby place it into the wider class of ideal fluid dynamics with advected quantities. In particular, the proof will identify the Poisson bracket in \eqref{short-note} as being defined over domain $\mc D$ on functionals of the \emph{dual space}\footnote{Dual with respect to the $L^2$ pairing on $\mc D$} $(f_1\circledS f_2)^*$ of the Lie algebra  $(f_1\circledS f_2)$ of semidirect-product symplectic transformations. Thus, EBC dynamics is understood as coadjoint motion generated by the semidirect-product action of the Lie algebra $(f_1\circledS f_2)$ on function pairs $(f_1;f_2)\in (f_1\circledS f_2)^*$. This proof of coadjoint motion also identifies the potential vorticity and buoyancy, $(\omega;b)$, as a semidirect-product momentum map.  

The Lie algebra commutator $[\,\cdot\,,\,\cdot\,]$ action for the adjoint (ad) representation of the action of the Lie algebra $f_1\circledS f_2$ on itself is defined by 
\begin{align}
{\rm ad}_{\big(\overline{f}_1;\overline{f}_2\big)}\big(f_1;f_2\big) 
=
\Big[ \big(f_1;f_2\big) , \big(\overline{f}_1;\overline{f}_2\big)\Big] 
:=
\Big( \big[ f_1,\overline{f}_1 \big] ; \big[ f_1,\overline{f}_2 \big] - \big[ \overline{f}_1,f_2 \big]\Big)\,,
\label{SDP-ad-action-fns}
\end{align}
where the commutator $[\,\cdot\,,\,\cdot\,]$ is given by the Jacobian of the functions $f_1$ and $\overline{f}_2$. For example,
\begin{align}
[ f_1,\overline{f}_2 ]:= J(f_1,\overline{f}_2)
\,,\label{SDP-commutator-fns}
\end{align}
which is also the commutator of symplectic vector fields. Thus, the adjoint (ad) action in \eqref{SDP-ad-action-fns} is the semidirect product Lie algebra action among symplectic vector fields. 

The definition in \eqref{SDP-ad-action-fns} of the semidirect product Lie algebra action among functions defined on the plane $\mathbb{R}^2$ enables the Lie-Poisson bracket \eqref{EBC-brkt3} for functionals of $(\omega,\Theta)$ defined on domain $\mc D$ to be identified with the coadjoint action of this Lie algebra. This is because the variational derivatives of such functionals live in the Lie algebra of symplectic vector fields on domain $\mc D$. Thus,
\begin{align}
\begin{split}
 \frac{\dd }{\dt} \mc{F}(\omega,\Theta)
 =
\Big\{\mc F\,,\,{\mc H}\Big\}(\omega,\Theta)
&=
- \Bigg\langle
\big(\omega ; \Theta\big) \,,\,\Bigg[ \bigg(\frac{\delta \mc F}{\delta \omega} ; \frac{\delta \mc H}{\delta \Theta}\bigg), 
\bigg(\frac{\delta \mc H}{\delta \omega} ;\frac{\delta \mc F}{\delta \Theta} \bigg)\Bigg] \Bigg\rangle
\\&=:
- \Bigg\langle \big(\omega ; \Theta\big)
\,,\,
{\rm ad}_{\big(\frac{\delta \mc H}{\delta \omega};\frac{\delta \mc H}{\delta \Theta}\big)}
\bigg(\frac{\delta \mc F}{\delta \omega};\frac{\delta \mc F}{\delta \Theta}\bigg)
 \Bigg\rangle
\\&=:
- \Bigg\langle {\rm ad}^*_{\big(\frac{\delta \mc H}{\delta \omega};\frac{\delta \mc H}{\delta \Theta}\big)}\big(\omega ; \Theta\big)
\,,\,
\bigg(\frac{\delta \mc F}{\delta \omega};\frac{\delta \mc F}{\delta \Theta}\bigg)
 \Bigg\rangle
\,,
\end{split}
\label{SDP-LPB-EBC}
\end{align}
in which the angle brackets $\langle\,\cdot\,,\,\cdot\,\rangle$ represent the pairing of the Lie algebra of functions $f_1\circledS f_2$ with its dual Lie algebra $(f_1\circledS f_2)^*$ via the $L^2$ pairing on the domain $\mc D$ as in equation \eqref{EBC-brkt3}. Consequently, we may rewrite the EBC equations in \eqref{EBC-eqns-def1} in terms of this coadjoint action and thereby reveal its geometric nature,
\begin{align}
\frac{\partial \big(\omega ; \Theta\big)}{\partial t}
= - \,{\rm ad}^*_{\big(\frac{\delta \mc H}{\delta \omega};\frac{\delta \mc H}{\delta \Theta}\big)}
\big(\omega ; \Theta\big)
= - \,{\rm ad}^*_{\big(\psi\,;\, (\Theta -z)\big)}
\big(\omega ; \Theta\big)
\,.
\label{SDP-EBC-LP-ad-star}
\end{align}
Thus, EBC dynamics is governed by the semidirect-product coadjoint action shown in equation \eqref{SDP-LPB-EBC} of symplectic vector fields acting on the momentum map $(\omega ; \Theta)$ defined in the dual space of potential vorticity and buoyancy functions . 
\smallskip

For completeness, we write the coadjoint (ad$^*$) representation of the (right) action of the Lie algebra $f_1\circledS f_2$ on its dual Lie algebra explicitly in terms of the Jacobian operator between pairs of functions as 
\begin{align}
\begin{split}
\Big\langle \big(\omega ; \Theta\big) \,,\,
{\rm ad}_{\big(\overline{f}_1;\overline{f}_2\big)}\big(f_1;f_2\big) 
\Big\rangle
&:= -\,
\Big\langle \big(\omega ; \Theta\big) \,,\,
\Big[ \big(f_1;f_2\big) , \big(\overline{f}_1;\overline{f}_2\big)\Big] 
\Big\rangle
\\
\hbox{By \eqref{SDP-ad-action-fns} and \eqref{SDP-commutator-fns}}\quad
&:= -\,
\Big\langle \big(\omega ; \Theta\big) \,,\,
\Big( J\big( f_1,\overline{f}_1 \big) ; J\big( f_1,\overline{f}_2 \big) 
- J\big( \overline{f}_1,f_2 \big)\Big)
\Big\rangle
\\
\hbox{Defining the SDP pairing}\quad
&:= -\,
\int_{\mc D} 
\omega \,J\big( f_1 , \overline{f}_1 \big)
+ \Theta \Big( J\big( f_1,\overline{f}_2 \big) - J\big( \overline{f}_1,f_2 \big) \Big)
\,\dd x \dd y
\\& = -\,
\int_{\mc D} 
\omega \,J\big( f_1,\overline{f}_1 \big)
+  \Big( \Theta J\big( f_1,\overline{f}_2 \big) + \Theta J\big( f_2 , \overline{f}_1 \big) \Big)
\,\dd x \dd y
\\\hbox{By integrating by parts}\quad
& = 
\int_{\mc D} 
f_1\,\Big( J\big( \omega ,\overline{f}_1 \big) +  J\big( \Theta  ,\overline{f}_2 \big)\Big) 
- f_2J\big( \Theta , \overline{f}_1 \big) 
\,\dd x \dd y
\\
\hbox{Upon the SDP pairing}\quad
&= 
\Big\langle 
\Big( J\big( \omega,\overline{f}_1 \big) + J\big( \Theta ,\overline{f}_2\big)  \,;\, - J\big( \overline{f}_1,\Theta \big)
\Big)
\,,\,
\big( f_1 ; f_2\big) 
\Big\rangle
\\
\hbox{By \eqref{SDP-LPB-EBC}}\quad
&=: 
\Big\langle {\rm ad}^*_{\big(\overline{f}_1;\overline{f}_2\big)}\big(\omega ; \Theta\big) \,,\,
\big(f_1;f_2\big) 
\Big\rangle
\quad\hbox{with}\quad
\big(\overline{f}_1;\overline{f}_2\big)
=\Big(\frac{\delta \mc H}{\delta \omega};\frac{\delta \mc H}{\delta \Theta}\Big)
\,.
\end{split}
\label{SDP-adstar-action-fns}
\end{align}
$\bullet$  The calculation in \eqref{SDP-adstar-action-fns} confirms that the bracket in \eqref{EBC-brkt3} satisfies the Jacobi identity, by being a linear functional of a Lie algebra commutator which satisfies the Jacobi identity. 

\noindent
$\bullet$  The calculation in \eqref{SDP-adstar-action-fns} identifies the bracket in \eqref{EBC-brkt3} as the Lie-Poisson bracket for functionals of $(\omega ; \Theta)$ defined on the dual $(f_1\circledS f_2)^*$ of semidirect-product Lie algebra $f_1\circledS f_2$, whose commutator is defined in equation  \eqref{SDP-ad-action-fns}. 
\end{proof}

\subsection{Casimir conservation laws for deterministic EBC flows}
We are now in a position to explain the conservation laws for the Hamilton matrix operator in \eqref{EBC-LPHam}. Namely, those conservation laws comprise Casimir functions $\mc{C}_{\Phi,\Psi}(\omega,\Theta)$ whose variational derivatives are null eigenvectors of the Hamilton matrix operator in \eqref{EBC-eqns-Ham} which defines the \emph{semidirect-product Lie-Poisson bracket} as
 \begin{align}
  \begin{split}
 \frac{\dd }{\dt} \mc{F}(\omega,\Theta)
 &=
\int_{\mc D}
\begin{bmatrix}
{\delta \mc F}/{\delta \omega} 
\\ 
{\delta \mc F}/{\delta \Theta} 
\end{bmatrix}^T
\begin{bmatrix}
J(\omega,\,\cdot\,) & J(\Theta,\,\cdot\,) 
\\ 
J(\Theta,\,\cdot\,)  & 0
\end{bmatrix}
\begin{bmatrix}
{\delta \mc H}/{\delta \omega} 
\\ 
{\delta \mc H}/{\delta \Theta} 
\end{bmatrix}
\dd x\dd z
=:
\Big\{\mc F\,,\,{\mc H}\Big\}(\omega,\Theta)
\\&=-
\int_{\mc D} 
\omega \,J\bigg(\frac{\delta \mc F}{\delta \omega},\frac{\delta \mc H}{\delta \omega}\bigg)
+ \Theta \Bigg[J\bigg(\frac{\delta \mc F}{\delta \Theta},\frac{\delta \mc H}{\delta \omega}\bigg)
- J\bigg(\frac{\delta \mc F}{\delta \omega},\frac{\delta \mc H}{\delta \Theta}\bigg)\Bigg] \dd x\dd z
\,.
\end{split}
\label{EBC-brkt3}
\end{align}
In particular, when $\mc{F}(\omega,\Theta)=\mc{C}_{\Phi,\Psi}(\omega,\Theta)$ we have 
\begin{align}
  \begin{split}
 \frac{\dd }{\dt} \mc{C}_{\Phi,\Psi}(\omega,\Theta)
 &=
\Big\{\mc{C}_{\Phi,\Psi}\,,\,{\mc H}\Big\}(\omega,\Theta)
\\&=-
\int_{\mc D} 
\omega \,J\bigg(\Psi(\Theta),\frac{\delta \mc H}{\delta \omega}\bigg)
+  \Bigg(\Theta J\bigg(\Phi'(\Theta)+\omega\Psi'(\Theta),\frac{\delta \mc H}{\delta \omega}\bigg)
- \Theta J\bigg(\Psi(\Theta),\frac{\delta \mc H}{\delta \Theta}\bigg)\Bigg) \dd x\dd z
\\&=-
\int_{\mc D} 
\frac{\delta \mc H}{\delta \omega}J\bigg(\omega ,\Psi(\Theta)\bigg)
+ \Bigg( \frac{\delta \mc H}{\delta \omega}J\bigg(\Theta,\Phi'(\Theta)+\omega\Psi'(\Theta)\bigg)
- \frac{\delta \mc H}{\delta \Theta}J\bigg(\Theta,\Psi(\Theta)\bigg)\Bigg) \dd x \dd y
\,.
\\&=-
\int_{\mc D} 
\frac{\delta \mc H}{\delta \omega}\Big(J\big(\omega ,\Theta\big)
+  J\big(\Theta,\omega\big)\Big)
\Psi'(\Theta) \dd x \dd y
\\& = 0 \quad\hbox{for all}\quad {\mc H}(\omega,\Theta)
\,.
\end{split}
\label{EBC-CasimirBrkt3}
\end{align}
This calculation demonstrates the conservation of $\mc{C}_{\Phi,\Psi}(\omega ,\Theta)$ in \eqref{EBC-LPHam} for an arbitrary Hamiltonian ${\mc H}(\omega,\Theta)$. These quantities are called \emph{Casimir functions} and as shown in \eqref{EBC-CasimirBrkt3} their conservation follows because their variational derivatives comprise null eigenvectors of the semidirect-product Lie-Poisson bracket in \eqref{EBC-CasimirBrkt3}. Now,  Lie-Poisson brackets generate coadjoint orbits, and coadjoint orbits lie on level sets of the bracket's Casimir functions. The requirement that EBC motion takes place on level sets of the Casimir functions limits the function space available to EBC solutions. In particular, the EBC solutions are restricted to stay on the same Casimir level set $\mc{C}_{\Phi,\Psi}(\omega ,\Theta)=const$ as that of their initial conditions.

\section{EBC critical-point equilibria: linear Lyapunov stability conditions}\label{appendix-B}   
\subsection{\bf Equilibrium conditions for EBC}
At equilibrium, equations \eqref{EBC-eqns-def3} imply that $\psi$, $\Theta$ and $\omega$ are all functions of $z$. Thus, we may denote the equilibrium solutions as $\psi_e(z)$, $\Theta_e(z)$, $\omega_e(z)$. 
Equilibrium solutions of the EBC equations have recently been characterised in \cite{WenCharlie-etal2020}.
Here we study the class of EBC equilibria which are critical points of a certain constrained energy. We then study 
their linear Lyapunov stability conditions using the energy-Casimir approach and determine their linear instability properties by deriving the EBC version of the Taylor-Goldstein equation for stratified incompressible planar fluid flow. \medskip

\begin{proposition}[Existence of critical-point equilibria]\label{CritPtCond}
Critical points of the sum $\mc{H}_C:=\mc{H}_{EBC} + \mc{C}_{\Phi,\Psi}$ with $\mc{H}_{EBC}$ and $\mc{C}_{\Phi,\Psi}$ defined in \eqref{EBC-eqns-Ham} and \eqref{eq:casimirsEBC} are equilibrium solutions of the EBC system in \eqref{EBC-eqns-def1}.
\end{proposition}
\begin{proof}
By \eqref{EBC-CasimirBrkt3} the Hamiltonian dynamics for an arbitrary functional $\mc{F}(\omega,\Theta)$ is given by
\begin{align*}
 \frac{\dd }{\dt} \mc{F}(\omega,\Theta) = \Big\{\mc F\,,\,{\mc H}_C\Big\}(\omega,\Theta)
\,.
\end{align*}
The Poisson bracket of the right-hand side of this dynamical equation vanishes at a critical point $(\omega_e,\Theta_e)$, where the variational derivitive is defined as
\begin{align*}
\delta {\mc H}_C\big|_{(\omega_e,\Theta_e)}:=\langle D{\mc H}_C\big|_{(\omega_e,\Theta_e)}\,,\,(\delta\omega,\delta\Theta)\rangle=0
\,.
\end{align*}
Consequently, critical points $\delta {\mc H}_C\big|_{(\omega_e,\Theta_e)}=0$ of the constrained Hamiltonian ${\mc H}_C$ are equilibrium solutions of the EBC system in \eqref{EBC-eqns-def1}. Namely, at such points one finds $ {\dd \mc{F}}/{\dt}|_{(\omega_e,\Theta_e)}=0$ for an arbitrary functional $\mc{F}$. 
\end{proof}

\begin{proposition}[Critical-point equilibrium conditions]\label{LPB-EBC2}
Critical points $(\omega_e,\Theta_e)$ of the sum $\mc{H}_C:=H_{EBC} + \mc{C}_{\Phi,\Psi}$  satisfy the following equilibrium flow conditions for solutions of the EBC system in \eqref{EBC-eqns-def1}.
\begin{align}
\psi_e(z)+\Psi(\Theta_e)= 0 
\quad\hbox{and}\quad
\Theta_e + \Phi'(\Theta_e) + \omega_e\Psi'(\Theta_e) = z
\,.
\label{CasimirEquilEqns}
\end{align}
\end{proposition}
\begin{proof}
By direct calculation
\begin{align*}
0 = 
\delta {\mc H}_C\big|_{(\omega_e,\Theta_e)} = \int_{\mc{D}} \big(\psi_e+\Psi(\Theta_e)\big)\delta \omega 
+  \big((\Theta_e - z) + \Phi'(\Theta_e) + \omega_e\Psi'(\Theta_e))\big)\delta \Theta  dxdz
\,.
\end{align*}
Vanishing of the coefficients of the independent variations $\delta \omega$ and $\delta \Theta$ then establishes the  critical point conditions of Proposition \ref{LPB-EBC2}. Then Proposition \ref{CritPtCond} implies that solutions satisfying these conditions are equilibria of the EBC system in \eqref{EBC-eqns-def1}.
\end{proof}

\begin{remark}[Flow properties of critical-point equilibria]\label{DeltaPhysInterp}
The equilibrium relations in \eqref{CasimirEquilEqns} imply some physical interpretations of the fluid equilibria. For example,  the  streamfunction-velocity relation in \eqref{EBC-eqns-def1} implies that 
\begin{align}
U_e(z) := -\,\frac{d\psi_e}{dz} = \frac{d\Psi({\Theta}_e(z))}{dz} = \Psi'({\Theta}_e)\frac{d{\Theta}_e}{dz}
\,,\quad
\omega_e(z) = - \,\frac{dU_e}{dz} = \frac{d^2\psi_e}{dz^2}
\quad\hbox{and}\quad 
\omega_e'(z) = - \,U_e''(z)
\,.
\end{align}
Thus, these critical-point equilibria are horizontal flows with vertical shear that can depend nonlinearly on the vertical coordinate, $z$, through the arbitrary functions of equilibrium buoyancy $\Phi({\Theta}_e)$ and $\Psi({\Theta}_e)$. Moreover, the nonlinear vertical variations of the vorticity and buoyancy are related by the second equilibrium relation in \eqref{CasimirEquilEqns}.  In particular, the $z$-derivative of the second equilibrium relation in \eqref{CasimirEquilEqns} implies that
\begin{align}
\frac{d}{dz}\Big(\Theta_e + \Phi'(\Theta_e) + \omega_e\Psi'(\Theta_e)\Big)
= \frac{d{\Theta}_e}{dz}\Big(1+ \Phi''(\Theta_e) + \omega_e\Psi''(\Theta_e)\Big) 
+ \frac{d\omega_e}{dz} \Psi'(\Theta_e)
= 1
\,.
\label{PhysEquilEqns1}
\end{align}
Therefore, one finds two relations between the arbitrary functions of equilibrium buoyancy $\Phi({\Theta}_e)$ and $\Psi({\Theta}_e)$ that relate them to the physical fluid properties of the equilibria, \eqref{PhysEquilEqns1}
\begin{align}
1+ \Phi''(\Theta_e) + \omega_e\Psi''(\Theta_e)
= \frac{ d\Theta_e/dz - U_e(z) d\omega_e/dz } { (d\Theta_e/dz)^2 }
= A(z)
 \quad\hbox{and}\quad 
 \Psi'({\Theta}_e) = \frac{U_e(z)}{d\Theta_e/dz} = B(z)
\,.
\label{PhysEquilEqns2}
\end{align}
Since $\omega_e'(z) = - \,U_e''(z)$, these two relations imply the formula
\begin{align}
A(z) + \frac{\omega_e'(z)}{{\Theta}_e'(z)}B(z) = 1/{\Theta}_e'(z)
\,.\label{PhysEquilEqns3}
\end{align}
\end{remark}

\noindent
\begin{proposition}[Negative-definite stability conditions]\label{LPB-EBC3}
The second variation \[\delta^2 {\mc H}_C\big|_{(\omega_e,\Theta_e)}:=\big\langle D^2{\mc H}_C\big|_{(\omega_e,\Theta_e)}\,,\,\big((\delta\omega,\delta\Theta),(\delta\omega,\delta\Theta)\big)\big\rangle\] will be negative definite, provided the following conditions are satisfied,
\begin{align}
A := \big(1+ \Phi''(\Theta_e) + \omega_e\Psi''(\Theta_e)\big) < 0
\quad\hbox{and}\quad
A/|\bm{k}|_{min}^2 + \big(\Psi'(\Theta_e)\big)^2 > 0
\,,
\label{NegDefCon}
\end{align}
where $|\bm{k}|_{min}^2$ is the eigenvalue of the Laplacian with the smallest magnitude in the flow domain $\mc{D}$.
\end{proposition}

\begin{lemma}[Flow characteristics of the negative-definite conditions in \eqref{NegDefCon}]\label{LPB-EBC3-flow}
From the flow equilibrium conditions in \eqref{PhysEquilEqns2}, the conditions \eqref{NegDefCon} for negative definiteness of the second variation may be expressed in terms of flow characteristics of the equilibrium solution as
\begin{align}
  {|\bm{k}|_{min}^2}U_e(z)^2  \  >\    U_e(z) \frac{d\omega_e}{dz} -  \frac{d\Theta_e}{dz}  \  >\  0
\,.
\label{NegDefConPhys}
\end{align}
\end{lemma}


\begin{proof}
By direct calculation
\begin{align}
\delta^2 {\mc H}_C\big|_{(\omega_e,\Theta_e)} 
= \int_{\mc{D}} (\delta\omega,\delta\Theta)
\begin{pmatrix} \Delta^{-1} & \Psi'(\Theta_e)
\\
\Psi'(\Theta_e) & 1+ \Phi''(\Theta_e) + \omega_e\Psi''(\Theta_e
\end{pmatrix}
\begin{pmatrix} \delta\omega \\ \delta\Theta \end{pmatrix}dxdz
\,.
\label{2nd-varH}
\end{align}
Next, one uses the Poincar\'e inequality to estimate
\begin{align}
\int_{\mc{D}} \big(\delta\omega(\Delta^{-1}\delta\omega)\big) dxdz 
\ge - \int_{\mc{D}} \frac{(\delta\omega)^2}{{|\bm{k}|_{min}}^2} dxdz
\,,
\label{Poinc-estimate}
\end{align}
where $|\bm{k}|_{min}^2$ is the minimum magnitude of Fourier wavenumber of the Laplacian operator in the flow domain $\mc{D}$. 
By applying the Cayley's theorem along with the Poincar\'e inequality estimate \eqref{Poinc-estimate} one finds that the following quadratic form is negative definite, provided the inequality conditions hold in \eqref{NegDefCon}, now written in terms of the flow equilibrium conditions in \eqref{PhysEquilEqns2}, 
\begin{align*}
\delta^2 {\mc {\wt{H}}}_C\big|_{(\omega_e,\Theta_e)} 
:= \int_{\mc{D}} (\delta\omega,\delta\Theta)
\begin{pmatrix} - 1/ |\bm{k}|_{min}^2 & \frac{U_e(z)}{d\Theta_e/dz}
\\
\frac{U_e(z)}{d\Theta_e/dz} & \frac{ d\Theta_e/dz - U_e(z) d\omega_e/dz } { (d\Theta_e/dz)^2 }
\end{pmatrix}
\begin{pmatrix} \delta\omega \\ \delta\Theta \end{pmatrix}dxdz
\,,
\end{align*}
and expressed in terms of physical variables in equation \eqref{NegDefConPhys} of the lemma .
\end{proof}

\subsection{Linearised deterministic EBC equations}
\begin{proposition}[Linearised EBC equations]\label{LPB-EBC6}
Linearising equations \eqref{EBC-eqns-def1} around the equilibrium $(\omega_e ; \Theta_e)$ satisfying the nonlinear relations in \eqref{CasimirEquilEqns} yields the equation
\begin{align}
\begin{split}
\frac{\partial \big(\delta\omega ; \delta\Theta\big)}{\partial t}
&= \frac{-1\,}{2} \,{\rm ad}^*{\bigg( \frac{\delta \big(\delta^2 {\mc H}_C\big|_{(\omega_e,\Theta_e)}\big)}{\delta (\delta\omega)}
;
\frac{\delta \big(\delta^2 {\mc H}_C\big|_{(\omega_e,\Theta_e)}\big)}{\delta (\delta\Theta)\big)}}\bigg)
\big(\omega_e,\Theta_e\big) 
\\
\frac{\partial }{\partial t} 
\begin{bmatrix}
\delta\omega \\ \delta\Theta
\end{bmatrix}
&=
\frac{-1\,}{2}
\begin{bmatrix}
J(\omega_e,\,\cdot\,) & J(\Theta_e,\,\cdot\,) 
\\ 
J(\Theta_e,\,\cdot\,)  & 0
\end{bmatrix}
\begin{bmatrix}
{\delta \big(\delta^2 {\mc H}_C\big|_{(\omega_e,\Theta_e)}\big)}/{\delta (\delta\omega)}
\\ 
{\delta \big(\delta^2 {\mc H}_C\big|_{(\omega_e,\Theta_e)}\big)}/{\delta (\delta\Theta)\big)}
\end{bmatrix}
\\&=
-
\begin{bmatrix}
J(\omega_e,\,\cdot\,) & J(\Theta_e,\,\cdot\,) 
\\ 
J(\Theta_e,\,\cdot\,)  & 0
\end{bmatrix}
\begin{bmatrix}
\Delta^{-1}\delta\omega + \Psi'(\Theta_e)\delta\Theta
\\ 
\Psi'(\Theta_e)\delta\omega + \big(1+ \Phi''(\Theta_e) + \omega_e\Psi''(\Theta_e)\big)\delta\Theta
\end{bmatrix}
\,,
\end{split}
\label{SDP-EBC-Lin}
\end{align}
where $\delta^2 {\mc H}_C\big|_{(\omega_e,\Theta_e)}$ is given in equation \eqref{2nd-varH}.

\end{proposition}

\begin{proof}
The proof proceeds by a functional Taylor expansion of the entries of the linear operator ${\rm ad}^*$ in equation \eqref{SDP-EBC-LP-ad-star} in the neighbourhood of the equilibrium solution. For details, see Appendix A of \cite{HMRW1985}.
\end{proof}

\subsection*{Linear stability analysis}
Upon defining $\nu:=\delta\omega$ with $\overline{\phi}=\delta \psi = \Delta^{-1}\nu$, as well as $\theta:=\delta\Theta$ and $h(\nu,\theta) = \delta^2 {\mc H}_C\big|_{(\omega_e,\Theta_e)}$, the linearised equation in \eqref{SDP-EBC-Lin} that holds in the neighbourhood of the equilibrium satisfying the critical point condition $0 = 
\delta {\mc H}_C\big|_{(\omega_e,\Theta_e)}$ can be rewritten as
\begin{align}
\frac{\partial }{\partial t} 
\begin{bmatrix}
\nu \\ \theta
\end{bmatrix}
&=
-
\begin{bmatrix}
J(\omega_e(z),\,\cdot\,) & J(\Theta_e(z),\,\cdot\,) 
\\ 
J(\Theta_e(z),\,\cdot\,)  & 0
\end{bmatrix}
\begin{bmatrix}
\delta h/ \delta \nu =
\overline{\phi}+ \Psi'(\Theta_e)\theta
\\ 
\delta h/ \delta \theta =
\Psi'(\Theta_e)\nu + \big(1+ \Phi''(\Theta_e) + \omega_e(z)\Psi''(\Theta_e)\big)\theta
\end{bmatrix}
:=
\big\{(\nu ; \theta)\,,\, h \big\} 
\,,
\label{SDP-EBC-LinEqn}
\end{align}
in which the physical meanings of the terms in the variational derivatives involving the functions $\Phi(\Theta_e)$ and $\Psi(\Theta_e)$ are related to the fluid properties of the equilibrium solutions in equations \eqref{PhysEquilEqns1}. \medskip

\begin{proposition}[Preservation of second-variation quadratic form by the linearised equation, \cite{HMRW1985}]\label{LPB-EBC7}$\,$\\
The linearised equation in \eqref{SDP-EBC-LinEqn} preserves the quadratic form $h(\nu,\theta) = \delta^2 {\mc H}_C\big|_{(\omega_e,\Theta_e)}$ arising from the second variation in an infinitesimal neighbourhood of the equilibrium solution $(\omega_e,\Theta_e)$ for which ${\mc H}_C$ has a critical point. 
\end{proposition}

\begin{proof}
The matrix operator in equation \eqref{SDP-EBC-LinEqn} is antisymmetric under pairing in $L^2$. Consequently, 
\[
\frac{dh}{dt} = \big\{h(\nu ,\theta)\,,\, h \big\} = 0\,.
\]
\end{proof}

\begin{proposition}[Linear Lyapunov stability theorem]\label{LPB-EBC4}
Equilibrium solutions $(\omega_e,\Theta_e)$ that satisfy the conditions for negative definiteness in Proposition \ref{LPB-EBC3} are linearly Lyapunov stable with respect to the metric given by \emph{minus} the second variation: $-\,\delta^2 {\mc H}_C\big|_{(\omega_e,\Theta_e)}>0$.
\end{proposition}

\begin{proof}
This result follows because the quadratic form $h(\nu,\theta) = \delta^2 {\mc H}_C\big|_{(\omega_e,\Theta_e)}$ is preserved 
in the neighbourhood of $(\omega_e,\Theta_e)$. Hence, it can act as a norm for expressing linear Lyapunov stability. 
\end{proof}

\begin{remark}[Brief summary of EBC equilibrium properties]\rm
At this stage, we have characterised the critical-point equilibrium solutions of the EBC dynamics in \eqref{EBC-eqns-def1} and we have derived sufficient conditions  for their linear Lyapunov stability with respect to the second-variation metric defined in the neighbourhood of an equilibrium solution. Namely, initial perturbations that start in an infinitesimal neighbourhood of $(\omega_e,\Theta_e)$ defined by the metric quadratic form $-\,\delta^2 {\mc H}_C\big|_{(\omega_e,\Theta_e)}>0$ will stay in that neighbourhood under the evolution of the linearised equations. 

Of course, violation of a sufficient condition for stability is not necessary for linear instability. However, the method for determining the sufficient conditions  for the linear Lyapunov stability of critical-point equilibrium solutions introduces the second-variation metric defined in the neighbourhood of an equilibrium solution which turns out to be the Hamiltonian for the linearised equations in the neighbourhood of a critical-point equilibrium solution. The induced Hamiltonian structure of the linearised equations now opens the opportunity to investigate the connections between the second-variation sufficient conditions for linear Lyapunov instability and the spectral sufficient conditions for linear instability. For this linear-instability endeavour, we will apply relationships among the flow characteristics for the critical-point equilibrium solutions determined so far to cast the linearised equations into a form that admits direct physical interpretation. Perhaps not unexpectedly, these equations will turn out to be variants of well-known equations in the vast literature of investigations of linear instability of EBC equilibria.

\end{remark}

\medskip

\begin{theorem}
In the Boussinesq approximation, the linearised EBC equations satisfy the following Taylor-Goldstein equation for stratified, inviscid Euler flow,
\begin{align}
\big( c + U_e(z)\big) (\phi''(z)-k^2 \phi(z) ) -\, \frac{\phi(z)}{c + U_e(z)}
+  U_e''(z) \phi(z) = 0
\,.\label{EBC-LinEqn-phi-thm}
\end{align}
\end{theorem}

\begin{proof}
The proof begins by using the relations in \eqref{PhysEquilEqns2} to rewrite the linear EBC equations in \eqref{SDP-EBC-LinEqn},
\begin{align}
\begin{split}
\frac{\partial }{\partial t} 
\begin{bmatrix}
\nu \\ \theta
\end{bmatrix}
&=
-
\begin{bmatrix}
J(\omega_e(z),\,\cdot\,) & J(\Theta_e(z),\,\cdot\,) 
\\ 
J(\Theta_e(z),\,\cdot\,)  & 0
\end{bmatrix}
\begin{bmatrix}
\overline{\phi}+ \frac{U_e(z)}{d\Theta_e/dz}\theta
\\ 
\frac{U_e(z)}{d\Theta_e/dz}\nu +  \frac{ d\Theta_e/dz - U_e(z) d\omega_e/dz } { (d\Theta_e/dz)^2 }\theta
\end{bmatrix}
\\&=
-
\begin{bmatrix}
J(\omega_e(z),\,\cdot\,) & J(\Theta_e(z),\,\cdot\,) 
\\ 
J(\Theta_e(z),\,\cdot\,)  & 0
\end{bmatrix}
\begin{bmatrix}
\overline{\phi}+ B(z)\theta
\\ 
B(z)\nu +  A(z)\theta
\end{bmatrix}
\,,
\end{split}
\label{SDP-EBC-LinEqn-redux}
\end{align}
where $A(z)$ and $B(z)$ are defined in equation \eqref{PhysEquilEqns2}, and $J(a,b):= -a_z b_x + a_x b_z$. Consequently, the linearised equations may be written as
\begin{align}
\begin{split}
\frac{\partial \nu}{\partial t}  &= - \Big( J\big(\omega_e(z),\,\overline{\phi}+ B(z)\theta\,\big) 
+ J\big(\Theta_e(z),\,B(z)\nu +  A(z)\theta\,\big)  \Big)
\\&= \omega'_e(z) \partial_x\big(\,\overline{\phi}+ B(z)\theta\,\big) 
+ \Theta'_e(z)\partial_x \big(\,B(z)\nu +  A(z)\theta\,\big) 
\\
\frac{\partial \theta}{\partial t}  &= - \Big( J\big({\Theta}_e(z),\,\overline{\phi}+ B(z)\theta\,\big) \Big)
={\Theta}'_e(z) \partial_x\big(\,\overline{\phi}+ B(z)\theta\,\big) 
\,.\end{split}
\label{EBC-LinEqn-AB}
\end{align}
For the linear instability analysis, we separate variables as
\begin{align}
\begin{split}
\overline{\phi} &= e^{ik (x-ct)}\phi(z)\,,
\\
\theta &= e^{ik (x-ct)}\chi(z)\,.
\end{split}
\label{EBC-LinEqn-theta-chi}
\end{align}
Consequently, the $\nu$ variable also separates, as
\begin{align}
\nu(x,z,t) = \Delta\overline{\phi} = \big(\phi''(z) - k^2\phi(z)\big)e^{ik (x-ct)}
\,.
\label{EBC-LinEqn-nu}
\end{align}
Solving for $\chi$ in  the second equation of \eqref{EBC-LinEqn-AB} now yields
\begin{align}
\chi(z) = \frac{-\,{\Theta}_e'(z)}{c + {\Theta}_e'(z)B(z)}\phi(z) = \frac{-\,{\Theta}_e'(z)}{c + U_e(z)}\phi(z)
\,.\label{EBC-LinEqn-chi}
\end{align}
Substituting the expressions in \eqref{EBC-LinEqn-theta-chi} for $\overline{\phi}$ and $\theta$ into the first equation of \eqref{EBC-LinEqn-AB} and canceling a common factor of $ik e^{ik (x-ct)}$ throughout yields
\begin{align}
\big( c + {\Theta}_e'(z)B(z)\big) (\phi''(z)-k^2 \phi(z) ) + \big[ A(z) + \omega_e'(z) B(z) \big] \chi(z) 
+ \omega_e'(z) \phi(z) = 0
\,.
\end{align}
By formulas \eqref{PhysEquilEqns2} and \eqref{PhysEquilEqns3} this becomes 
\begin{align}
\big( c + U_e(z)\big) (\phi''(z)-k^2 \phi(z) ) -\, \frac{\phi(z)}{c + U_e(z)}
+  U_e''(z) \phi(z) = 0
\,,\label{EBC-LinEqn-phi1}
\end{align}
as required.
\end{proof}
\begin{remark}\rm
See \cite{AH1987,AHMR1986,HMRW1985} and references therein for discussions of the Taylor-Goldstein equation 
for the stratified Euler fluid equation in the plane. 
\end{remark}




{}


\begin{thebibliography}{}

\bibitem{AH1987} 
Abarbanel, H.D.I. and Holm, D.D., 1987.
Nonlinear Stability of Inviscid Flows in Three 
Dimensions: Incompressible Fluids and  Barotropic Fluids, 
{\it Phys. Fluids} {\bf 30}, 3369--3382.
\url{https://doi.org/10.1063/1.866469}

\bibitem{AHMR1986}
Abarbanel, H.D.I., Holm, D.D., Marsden, J.E. and Ratiu, T.S., 1986.
Nonlinear Stability Analysis of Stratified Ideal Fluid Equilibria, 
{\it Phil Trans. Roy. Soc. (London) A} 
{\bf 318} 349--409.
\url{https://doi.org/10.1098/rsta.1986.0078}

\bibitem{Alonso-Oran-bDLeon-JNLS2020}
Alonso-Or\'an, D. and Bethencourt de Le\'on, A., 2020. On the well-posedness of stochastic Boussinesq equations with transport noise. 
Journal of Nonlinear Science, 30(1), pp.175-224. \url{https://doi.org/10.1007/s00332-019-09571-2}

\bibitem{Alonso-Oran-etal-JStaPhys} 
Alonso-Or\'an, D.,  Bethencourt de Le\'on, A., Holm, D.D.,  and Takao, S., 2020.
Modelling the Climate and Weather of a 2D Lagrangian-Averaged Euler--Boussinesq Equation with Transport Noise
J. Stat. Phys.  179, 1267-1303.
\url{https://doi.org/10.1007/s10955-019-02443-9}

\bibitem{BKM1984}
Beale, J.T., Kato, T. and Majda, A., 1984. Remarks on the breakdown of smooth solutions for the 3-D Euler equations. Communications in Mathematical Physics, 94(1), pp.61-66.

\bibitem{cotter2018modelling} 
Cotter, C., Crisan, D., Holm, D.D., Pan, W., and Shevchenko, I., 2020. Modelling uncertainty using stochastic transport noise in a 2-layer quasi-geostrophic model, Found. Data Sci. 2(2), 173.
\url{https://doi.org/10.3934/fods.2020010}

\bibitem{cotter2019numerically}
Cotter, C.J., Crisan, D., Holm, D.D., Pan, W., and Shevchenko, I., 2019. Numerically modeling stochastic Lie transport in fluid dynamics, Multiscale Model. Simul. 17(1), 192-232. \url{https://doi.org/10.1137/18M1167929}

\bibitem{CotterGottwaldHolm2017} 
Cotter, C.J., Gottwald, G.A. Holm, D.D. [2017].   
Stochastic partial differential fluid equations as a diffusive limit of deterministic Lagrangian multi-time dynamics
Proc Roy Soc A, Vol 473 page 20170388 
\url{http://dx.doi.org/10.1098/rspa.2017.0388}

\bibitem{CHK-2022}
Crisan, D. Holm, D.D. and Korn, P., 2022.
Hasselmann’s Paradigm for Stochastic Climate Modelling based on Stochastic Lie Transport.
\url{https://arxiv.org/pdf/2205.04560.pdf}

\bibitem{CHLN2022}
Crisan, D., Holm, D.D., Leahy, J.M, and Nilssen, T., 2022. 
Variational principles for fluid dynamics on rough paths.
\textit{Adv in Math.} 404 108409. 
\url{https://doi.org/10.1016/j.aim.2022.108409}.  Preprint at arXiv:2004.07829.

\bibitem{CHLMP2021}
Crisan, D., Holm, D.D., Luesink, E., Mensah, P.R. and Pan, W., 2021. 
Theoretical and  computational analysis of the thermal quasi-geostrophic model.
\url{https://arxiv.org/pdf/2106.14850.pdf}

\bibitem{Landau}
Donev, A., Vanden-Eijnden, E., Garcia, A. and Bell, J., 2010.
On the accuracy of finite-volume schemes for fluctuating hydrodynamics.
Commun. Appl. Math. and Comput. Sci.  5 (2) 149-197.
\url{https://doi.org/10.2140/camcos.2010.5.149}

\bibitem{Doering1994}
Doering, C.R., Horsthemke, W. and Riordan, J., 1994. 
Nonequilibrium fluctuation-induced transport. Physical review letters, 72(19), p.2984.

\bibitem{Farhat-etal-RBCdata2020} 
Farhat, A., Johnston, H., Jolly, M. et al., 2020. Assimilation of Nearly Turbulent Rayleigh-B\'enard Flow 
Through Vorticity or Local Circulation Measurements: A Computational Study. 
J. Sci. Comput. 77, 1519-1533. \url{https://doi.org/10.1007/s10915-018-0686-x}

\bibitem{FlandoliPappa2022} 
Flandoli, F., Pappalettera, U. From additive to transport noise in 2D fluid dynamics. Stoch PDE: Anal Comp (2022). 
\url{https://doi.org/10.1007/s40072-022-00249-7}

\bibitem
{Foias-etal-RBCanalysis1987}
Foias, C., Manley, O., Temam, R., 1987. Attractors for the B\'enard problem: existence and physical bounds on their fractal dimension. Nonlinear Anal. Theory Methods Appl. 11(8), 939-967.

\bibitem{Frisch1995}
U. Frisch, Turbulence. The Legacy of A.N. Kolmogorov
(Cambridge University Press, Cambridge, 1995).
%
%
%
%
%
\bibitem{GHL2019}
Geurts, B.J., Holm, D.D. and Luesink, E., 2019. 
Lyapunov Exponents of Two Stochastic Lorenz 63 Systems. 
Journal of Statistical Physics, pp.1-23.
\url{https://doi.org/10.1007/s10955-019-02457-3}

\bibitem{Holm2015}
Holm, D.D., 2015.
Variational principles for stochastic fluid dynamics.
\textit{Proc Roy Soc A}, 471: 20140963.  
\url{http://dx.doi.org/10.1098/rspa.2014.0963}
%
\bibitem{HH2021}
Holm, D.D., and Hu, R., 2021.
Stochastic effects of waves on currents in the ocean mixed layer.
Journal of Mathematical Physics 62, 073102.
\url{https://aip.scitation.org/doi/10.1063/5.0045010}
%
%
\bibitem{HMR1998} 
Holm, D.D., Marsden, J.E. and Ratiu, T.S., 1998.
The Euler--Poincar\'e equations and semidirect products
with applications to continuum theories,
{\it Adv. in Math.}, {\bf 137} 1-81,
\url{https://doi.org/10.1006/aima.1998.1721}

\bibitem{HMRW1985}
Holm, D.D., Marsden, J.E., Ratiu, T.S., and Weinstein, A., 1985. Nonlinear stability
of fluid and plasma equilibria. Phys. Rep. 123, 1-116.
\url{https://doi.org/10.1016/0370-1573(85)90028-6}

\bibitem{HuPatching2022} 
Hu, R. and Patching, S., 2022. 
Variational Stochastic Parameterisations and their Applications to Primitive Equation Models.
Preprint at arXiv:2202.04404

\bibitem{Kraichnan1968}
Kraichnan, R.H., 1968. 
Small-scale structure of a scalar field convected by turbulence. 
The Physics of Fluids, 11(5), pp.945-953.

\bibitem{Kraichnan1994} 
Kraichnan, R.H., 1994. 
Anomalous scaling of a randomly advected passive scalar. 
Phys. Rev. Lett., 72(7), p.1016.

\bibitem{Kovalevsky2020}
Kovalevsky, D.V., Bashmachnikov, I.L. and Alekseev, G.V., 2020. 
Formation and decay of a deep convective chimney. Ocean Modelling, 148, p.101583.

\bibitem{McKean1966}
McKean,  H. P. Jr., [1966] A class of Markov processes associated with nonlinear parabolic equations. Proc.
Nat. Acad. Sci. U.S.A., 56:1907-1911.

\bibitem{Saltzman1962}
Saltzman, B., 1962.
Finite Amplitude Free Convection as an Initial Value Problem-I. Journal of the Atmospheric Sciences, 19, 329-342.\\
\url{https://doi.org/10.1175/1520-0469(1962)019<0329:FAFCAA>2.0.CO;2}

\bibitem{Temam-RBCanalysis1997}
Temam, R., 1997. \textit{Infinite Dimensional Dynamical Systems in Mechanics and Physics}. Applied Mathematical Sciences, vol. 68, 2nd edn. Springer, New York.

\bibitem{JL+Horton1996}
Thiffeault, J.-L. and Horton, W., 1996.
Energy-conserving truncations for convection with shear flow.
Physics of Fluids 8, 1715 \url{https://doi.org/10.1063/1.868956}

\bibitem{WenCharlie-etal2020}
Wen, B., Goluskin, D., LeDuc, M., Chini, G.P. and Doering, C.R., 2020. 
Steady Rayleigh-B\'enard convection between stress-free boundaries. 
Journal of Fluid Mechanics, 905, R4.\\ 
\url{https://doi:10.1017/jfm.2020.812}

%
%
\end{thebibliography}
\end{document}